\documentclass[prodmode,acmtkdd]{acmsmall}

\usepackage[hidelinks]{hyperref}
\usepackage{subfig}
\usepackage{indentfirst}
\usepackage{amsfonts, amssymb, amsmath, amscd}
\usepackage{euscript}
\usepackage{algorithm}
\usepackage[noend]{algpseudocode}

\hypersetup {
  colorlinks = false,
  citecolor = black,
  filecolor = black,
  linkcolor = black,
  urlcolor = black,
  pdfauthor = {Kostyantyn Demchuk and Douglas J. Leith},
  pdftitle = {Kostyantyn Demchuk and Douglas J. Leith - A Fast Minimal
  Infrequent Itemset Mining Algorithm},
  pdfsubject = {A Fast Minimal Infrequent Itemset Mining Algorithm},
  pdfkeywords = {infrequent itemset mining, rare patterns, data mining,
  algorithm, breadth-first, frequency-based analysis, $k$-anonymity,
  performance, load balancing, applied mathematics, java, c++},
  pdfpagemode = UseNone
}

\algrenewcommand{\algorithmicforall}{\textbf{foreach}}
\def\ForEach{\ForAll}
\algrenewcommand\algorithmicindent{1.0em}
\algnewcommand{\LineComment}[1]{\State \(\triangleright\) #1}
\newcommand\NoThen{\renewcommand\algorithmicthen{}}
\newcommand\ReThen{\renewcommand\algorithmicthen{\textbf{then}}}

  \title{A Fast Minimal Infrequent Itemset Mining Algorithm\thanks{This work
  was supported by an IBM PhD Fellowship and by Science Foundation Ireland
  under Grant No. 11/PI/1177.}}

  \author{KOSTYANTYN DEMCHUK and DOUGLAS J. LEITH
  \affil{NUI Maynooth, Ireland}}

  \begin{abstract}
    A novel fast algorithm for finding quasi identifiers in large datasets is
    presented. Performance measurements on a broad range of datasets
    demonstrate substantial reductions in run-time relative to the state of the
    art and the scalability of the algorithm to realistically-sized datasets up
    to several million records.
  \end{abstract}


  \acmformat{Kostyantyn Demchuk and Douglas J. Leith. 2014. Fast Minimal
  Infrequent Itemset Mining Algorithm}

  \makeatletter
  \g@addto@macro{\maketitle}{\@thanks}
  \makeatother

\begin{document}

  \maketitle

\section{Introduction}

  In this paper we introduce a new algorithm, called Kyiv, for finding all
  minimal attribute combinations occurring with less than a specified frequency
  within a data set. On realistic data sets this algorithm is demonstrated to
  be considerably faster than state of the art algorithms.

  One application of this algorithm is in statistical disclosure control
  \cite{suda2,minit,merrett04,eurostat07,sdcmicro14}. In statistical disclosure
  control the released data, for example census microdata, is required to be
  suitably anonymised. Of particular concern is the removal of
  quasi-identifiers \emph{i.e.} a subset of attribute values that can uniquely
  identify one or more entries in a data set. Even apparently innocuous data
  can act as a quasi-identifier when multiple values are combined together. For
  example, the seminal study of Sweeney \cite{sweeney02} showed that $87\%$ of
  the US population are uniquely identified by the three attributes gender, zip
  code and date of birth and demonstrated the use of this fact to de-anonymise
  published health data. It is therefore of fundamental interest to enumerate
  those combinations of entries within a dataset which occur either uniquely or
  sufficiently infrequently.

  Other applications of our algorithm include rare itemset mining
  \cite{koh05,tsang11,szathmary10,tsang13}. In rare itemset mining the aim is
  to discover unusual, but informative, relationships between entries in a data
  set. This is in contrast to frequent itemset mining where the interest is in
  discovering relationships which are common within a data set. Rare but
  interesting items might for example include adverse drug reactions within
  medical data \cite{ji13} and attacker intrusion within network data
  \cite{rahman08,luna10,hommes12} \emph{etc}. Since rare items are, by
  definition, infrequent, a direct approach to discovery is to enumerate the
  infrequent items and then search for informative relationships, \emph{e.g.}
  those which are of sufficiently high confidence, within this enumerated set.

  The main contributions of the paper are as follows. We introduce a new
  algorithm for minimal infrequent itemset mining, in both sequential and
  parallel form. The main practical contribution is the speed up of almost two
  orders of magnitude offered by the proposed algorithm on datasets of
  realistic complexity. Since execution time is currently the primary
  bottleneck in finding minimally infrequent itemsets, this is a significant
  step forward. The main algorithmic novelty (from which the speed up arises)
  is that by an appropriate choice of data structures and algorithmic
  formulation the support item test for minimality can be performed in a hugely
  more efficient manner (essentially with zero cost) than previously possible.
  A second algorithmic contribution lies in the parallel implementation. Unlike
  some previous approaches, the proposed approach elegantly allows the work
  load of parallel threads to be balanced so as to be approximately the same.
  This means that no single thread becomes the performance bottleneck and
  therefore ensures better scalability. We note that the speed up in execution
  time comes at the cost of much higher memory usage. However, since available
  memory size continues to grow year on year while processor speed has largely
  stagnated in many practical applications this trade-off of memory for speed
  is a favourable one. The new algorithm design is underpinned by new analytic
  results, the main analytic contribution lying in Lemma \ref{lem} and
  Corollary \ref{cor}. We present experimental measurements evaluating the
  performance of the proposed algorithm on a range of synthetic and application
  datasets, and compare this against the performance of the popular algorithm
  MINIT \cite{minit} and of the recently proposed MIWI Miner algorithm
  \cite{cagliero13}.

  \subsection{Motivating Example}

  \begin{table}
    \centering
    \begin{tabular}{| c | p{5.5cm} | p{2cm}  | l |}
    \hline
    ID & Query & Date & Link Clicked \\
    \hline
    3302 & uterine bleeding and coumadin & 2006-03-23 11:23:35 &
    www.nlm.nih.gov \\
    3302 & children who have died from moms postpartum depression & 2006-03-24
    15:41:21 & www.cbsnews.com \\
    6993 & american heart association & 2006-03-23 18:29:34 &
    www.americanheart.org \\
    6993 & high blood pressure & 2006-03-23 18:37:10 & \\
    7005 & notice of demand to pay judgment form & 2006-03-21 18:49:01 &
    www.sba.gov \\
    7005 & free personal credit report & 2006-03-20 11:26:42 & www.experian.com
    \\
    4417749 & shadow lake subdivision gwinnett county georgia & 2006-04-24
    21:48:01 & \\
    4417749 & jarrett t. arnold eugene oregon & 2006-03-23 21:48:01 &
    www2.eugeneweekly.com \\
    \hline
    \end{tabular}
    \caption{Extracts from AOL web search dataset}
    \label{tab:one}
  \end{table}

  In 2006 AOL released web search log data in which user identities had been
  concealed (replaced by unique identity numbers) but other data was left
  unchanged. Table \ref{tab:one} presents some entries from this AOL data set.
  It can be seen that the search queries and pages clicked are potentially
  sensitive in nature and it was further demonstrated that de-anonymisation of
  users was possible \emph{e.g.} that user \#4417749 was Thelma Arnold
  \cite{nyt}.

  We consider quasi-identifiers within the search data for the first 65,517
  users in more detail. These users carried out 3,558,412 searches using
  1,216,655 distinct queries. Of these queries, 736,967 occur only once within
  the data set and so are potential quasi-identifiers. Restricting
  consideration to the first three words of each query reduces the number of
  unique queries to 617,510, while restricting to the first two words reduces
  this to 488,138 and restricting to the first word only yields 276,074 unique
  queries. Hence, it can be seen that simply truncating the search queries is
  not sufficient to prevent a large number of the search queries from acting as
  quasi-identifiers.

  One simple and direct approach to masking these unique queries is to group
  unique queries together into sets of queries where each set consists of $k$
  unique queries, $k$ being a design parameter. In the data set we now replace
  the query by a reference to the set containing the query. In this way it is
  ensured that every query value in the modified data set occurs at least $k$
  times within the data set. We performed this data transformation on the AOL
  data using a value $k = 5$. In addition, we performed a similar
  transformation to the web page clicked by a user following a query, also with
  $k = 5$. After these changes each query value and each web page clicked value
  occurs at least $k = 5$ times within the modified data set. Nevertheless,
  when this query value is combined with the web page clicked value 586,698 of
  these pairs are still unique within the modified data set. In the unmodified
  data set there are 1,030,387 unique pairs, so the grouping of query of page
  clicked values has also reduced the number of unique pairs. However, in view
  of the large value of unique pairs it is evidently not sufficient to just
  consider individual entries but rather it is also necessary to consider
  combinations of entries when anonymising a data set.

  The difficulty with considering combinations of entries is that the number of
  combinations to be tested grows combinatorially and so in realistically sized
  data sets highly efficient algorithms are needed to test even combinations of
  3 or 4 entries. One solution to this combinatorial growth is to use sampling.
  For example, a subset of entries may be drawn uniformly at random from the
  full dataset, the number of attribute combinations occurring with less than a
  specified frequency within this subset determined and then this information
  is statistically extrapolated to the full dataset. Sampling reduces the
  computational burden but also carries the obvious risk of missing
  infrequently occurring entries. More efficient algorithms allow consideration
  of larger samples and so potentially significantly reduce this risk.

  Note that the set of unique or sufficiently infrequently occuring
  combinations of items within a data set is useful not just for verifying that
  restrictions on quasi-identifiers are respected by a data set but, when
  quasi-identifiers are present, this set is also useful as input to tools such
  as that in \cite{sigmod05} for modifying the data that require prior
  knowledge of the fields which act as quasi-identifiers. In the above AOL
  example the set of unique combinations is the precisely set of elements from
  which grouped values need to be constructed.

\section{Related Work} \label{sec:related}


  The first algorithm for unique itemset mining (the extreme case of infrequent
  itemset mining) appears to be SUDA (special unique detection algorithm)
  proposed in \cite{suda}. This was followed shortly afterwards by the
  development of the SUDA2 algorithm \cite{suda2-jnl,suda2}, which uses a
  recursive depth-first search approach to generate candidate itemsets from the
  database of interest (thus every candidate itemset exists in the database)
  and then efficiently tests these for uniqueness and minimality. SUDA2 lends
  itself readily to parallelisation by allocating disjoint subtrees to
  different threads which then carry out a depth-first search on the subtree.
  However, the work allocated amongst threads may be imbalanced depending on
  the size and complexity of the subtree assigned to a thread, leading to
  performance being constrained by the slowest running thread. A number of
  mitigating strategies are therefore summarised in \cite{haglin09}. SUDA2 is
  available in the sdcMicro package for R \cite{sdcmicro} and is essentially
  the state-of-the-art algorithm in this area, being used by the UK and
  Australian national statistics offices \cite{haglin09} and supported by IHSN
  (International Household Survey Network).

  Early work on infrequent (rather than only minimal) itemset mining initially
  made use of variants of the Apriori algorithm for frequent itemset mining,
  see \cite{dong07} and references therein, but quickly moved on to algorithms
  specifically tailored to the infrequent mining task. Almost simultaneously
  three specialised infrequent itemset algorithms were proposed by
  \cite{zhou07}, \cite{szathmary07} and \cite{minit}. In \cite{zhou07} a hash
  based scheme referred to as HBS is proposed to mine association rules among
  rare items, involving a direct search of item sequences contained in a
  database with pruning based on frequency. In \cite{szathmary07} an algorithm
  referred to as ARIMA (a rare itemset miner algorithm) is proposed, and later
  refined in \cite{szathmary12} by the addition of a depth-first search to
  exclude frequent itemsets. In \cite{minit} the MINIT (minimal infrequent
  itemsets) algorithm is proposed. MINIT uses a recursive depth-first search
  with pruning, similarly to the SUDA2 algorithm developed by the same group,
  and is often used as the baseline algorithm against which the performance of
  other infrequent mining algorithms is compared. In \cite{troiano09,troiano13}
  a breadth-first algorithm, Rarity, aiming at finding not necessarily minimal
  infrequent itemsets, is introduced. Whereas other algorithms start from small
  itemsets and increase the size as they search, Rarity takes the opposite
  approach and proceeds from large itemsets to smaller ones (referred to in
  \cite{troiano09,troiano13} as a top-down strategy). In \cite{gupta11} a
  pattern-growth recursive depth-first approach is proposed for minimal
  infrequent itemset mining and two algorithms called IFP\_min and IFP\_MLMS
  (multiple level minimum support) are introduced. It is observed that there
  exists a frequency threshold below which MINIT generally outperforms IFP\_min
  and above which IFP\_min outperforms MINIT. IFP\_min is also observed to
  outperform MINIT for large dense datasets. Recently, \cite{cagliero13}
  extends consideration to the more general task of discovering infrequent
  weighted itemsets (IWI) and introduces an algorithm called MIWI (minimal IWI)
  Miner. When a weighting of unity is associated with every itemset then this
  reduces to the infrequent itemset mining problem. For the datasets
  considered, MIWI Miner is demonstrated to significantly outperform MINIT for
  infrequent itemset mining. However, it is worth noting that the performance
  comparison in \cite{cagliero13} is made only for a small number of datasets.

\section{Preliminaries} \label{sec:prelim}

  A dataset $A$ is a table with $n$ rows and $m$ columns. The columns in this
  table contain categorical or finite range continuous data (such as age,
  income, zip code \emph{etc}). Formally,

  \begin{definition}[Item]
    An item $a$ is a triple $(v, j_a, R_a)$ in $A$, where $v \in \mathbb{N}$ is
    its value, $j_a \in \{ 1, \dots, m \}$ is the column of $A$ containing $v$,
    and $R_a \subseteq \{ 1, \dots, n \}$ is the set of $A$ rows in which the
    item appears.
  \end{definition}

  Note that the column in which it appears distinguishes an item, the same
  value appearing in two different columns being treated as two different
  items. This is in line with previous work on infrequent itemset mining. Also
  observe that we consider items with values from the field of positive integer
  (natural) numbers $\mathbb{N}$, but since any countable set can be mapped on
  to the integers this restriction is mild (while real values are excluded,
  finite-precision values are admissible).

  Let $I_A$ denote the set of all items in $A$. We define the frequency and
  uniformity of items in the natural way, as follows:

  \begin{definition}[Frequency]
    An itemset $I \subseteq I_A$ is a set of items. A $k$-itemset refers to an
    itemset of cardinality $k$. We let $R_I = \bigcap_{a \in I} R_a$ denote the
    set of rows in which all items of $I$ appear, and we refer to $|R_I|$ as
    the frequency of itemset $I$.
  \end{definition}

  \begin{definition}[$\tau$-Infrequency]
    An item $a \in I_A$ is $\tau$-infrequent if it has frequency less than
    $\tau$ \emph{i.e.} $|R_a| \le \tau$ and so the item occurs in $\tau$ or
    fewer rows of the dataset. We let $r_{A,\tau} \subseteq I_A$ denote the
    set of $\tau$-infrequent items in $I_A$. Unless otherwise stated, we
    confine consideration to $\tau$ values less than $n$, since trivially all
    elements of the dataset are $n$-infrequent. Usually $0 < \tau \ll n$.
  \end{definition}

  \begin{definition}[Uniqueness]
    An item $a \in I_A$ is unique if it is $1$-infrequent. That is, $|R_a| = 1$
    and so the item occurs in dataset $A$ in exactly one row. We let $\delta_A
    \subseteq I_A$ denote the set of unique items in $I_A$.
  \end{definition}

  \begin{definition}[Uniformity]
    Let $B \subseteq \{ 1, \dots, n \}$ be a subset of row indices from dataset
    $A$, and let $I_B = \{ a \in I_A : R_a \cap B \ne \emptyset \}$. An item
    $a$ is said to be uniform in $I_B$ if $|R_a \cap B| = |B|$. That is, item
    $a$ occurs in every row of subtable $B$. We let $U_A = \{ a \in I_A : |R_a|
    = n \}$ denote the set of uniform items in $I_A$.
  \end{definition}

  \begin{example} \label{ex1}
    For dataset
    \begin{align*}
      A &= 
      \begin{bmatrix}
        1 & 2 & 3 & 4 \\
        1 & 2 & 7 & 4 \\
        1 & 6 & 3 & 4 \\
        5 & 2 & 3 & 4
      \end{bmatrix}
    \end{align*}
    we have
    \begin{align*}
      I_A &= \{ (1, 1, \{1, 2, 3\}), (2, 2, \{1, 2, 4\}), (3, 3, \{1, 3, 4\}),
      (4, 4, \{1, 2, 3, 4\}), \\
      & \hspace{0.65cm} (5, 1, \{4\}), (6, 2, \{3\}), (7, 3, \{2\}) \}. \\
      \delta_A &= \{ (5, 1, \{4\}), (6, 2, \{3\}), (7, 3, \{2\}) \}. \\
      U_A &= \{ (4, 4, \{1, 2, 3, 4\}) \}. \\
      r_{A,\tau} &= \left\{
      \begin{array}{rl}
        \emptyset & \mbox{if $\ \tau \le 0$} \\
        \delta_A & \mbox{if $\ 0 < \tau < 3$} \\
        I_A \setminus U_A & \mbox{if $\ \tau = 3$} \\
        I_A & \mbox{if $\ \tau > 3$}
      \end{array}. \right.
    \end{align*}
    \qed
  \end{example}

  \begin{definition}[$\tau$-Infrequent and Minimal Itemsets]{\ \newline}
    \label{def:mri}
    An itemset $I \subseteq I_A$ is $\tau$-infrequent and minimal if:
    \begin{enumerate}
      \item $\tau$-Infrequency: $|R_I| \le \tau$;
      \item Minimality: $|R_S| > \tau$ $\forall S \subset I$, $S \ne
      \emptyset$.
    \end{enumerate}
    When $\tau = 1$ we refer to the $\tau$-infrequent and minimal itemsets as
    being the \emph{unique and minimal itemsets} and in this case we often drop
    any $\tau$ subscripts to streamline notation.
  \end{definition}

  Note that to establish minimality in Definition \ref{def:mri} it is only
  necessary to test that $|R_S| > \tau$ for sets $S \subset I$ of size
  $|I| - 1$ since $R_{S^\prime} \supseteq R_S \ \forall S^\prime \subset S$.
  These $|I| - 1$ subsets are referred to as the \emph{support itemsets} of
  $I$. Notice also that itemsets of size $1$ (items) are trivially minimal.

  We denote the set of all unique and minimal itemsets by $\mathcal{I}_A
  \subseteq 2^{I_A}$ and the set of all $\tau$-infrequent and minimal itemsets
  by $\mathcal{I}_{A,\tau} \subseteq 2^{I_A}$, where $2^{I_A}$ denotes the set
  of all subsets of $I_A$. We use calligraphic script to indicate that
  $\mathcal{I}_A$ is a set of sets (similarly for $\mathcal{I}_{A,\tau}$) and
  to distinguish it from the set of items $I_A$. Notice that $\mathcal{I}_{A,
  \tau} = \mathcal{I}_A$ when $\tau = 1$.

\section{Minimal Infrequent Itemset Mining} \label{sec:algo}

  In this section we introduce a new algorithm for efficiently finding all of
  the $\tau$-in\-frequent and minimal $k$-itemsets up to a user specified size
  $k_{max}$, $1 \le k \le k_{max} \le m$ and frequency threshold $\tau > 0$.

  \subsection{Pre-processing}

  We begin by observing that uniform items $u \in U_A$ can be deleted from
  $I_A$ as they cannot form a minimal $\tau$-infrequent itemset (if $u \in I$
  and $|R_I| \le \tau$ then $|R_S| = |R_I| \ngtr \tau$ for $S = I \setminus
  \{ u \}$). Further, the set of $\tau$-infrequent individual items $r_{A,
  \tau}$ can be readily identified by direct search. The remaining set of
  non-uniform and non-$\tau$-infrequent items $I_{A,\tau}^\prime = I_A
  \setminus U_A \setminus r_{A,\tau}$ can be partitioned into sets $L_{A,\tau}$
  and $\bar{L}_{A,\tau} = I_{A,\tau}^\prime \setminus L_{A,\tau}$ such that (i)
  $R_a \ne R_b$ $\forall a, b \in L_{A,\tau}$, (ii) $\forall c \in \bar{L}_{A,
  \tau}$ there exists $d \in L_{A,\tau}$ with $R_c = R_d$. That is, within set
  $L_{A,\tau}$ no items share the same set of rows. This partitioning can be
  achieved in the obvious way. Namely, for any set of items in $I_{A,
  \tau}^\prime$ which share the same set of rows, add one of these items to
  $L_{A,\tau}$ and the rest to $\bar{L}_{A,\tau}$. Revisiting Example
  \ref{ex1}, we have $L_{A,\tau} = \{ (1, 1, \{1, 2, 3\}), (2, 2, \{1, 2, 4\}),
  (3, 3, \{1, 3, 4\}) \}$ for $0 < \tau < 3$.

  The partitioning into $L_{A,\tau}$ and $I_{A,\tau}^\prime \setminus L_{A,
  \tau}$ possesses the following useful property:
  \begin{proposition} \label{prop:partition}
    Let $W \subseteq L_{A,\tau}$ be a minimal $\tau$-infrequent itemset. Let
    $w^\prime \in I_A \setminus L_{A,\tau}$ with $R_w = R_{w^\prime}$ for some
    $w \in W$. Then $W \setminus \{ w \} \cup \{ w^\prime \}$ is also a minimal
    $\tau$-infrequent itemset.
  \end{proposition}
  \begin{proof}
    Since $W$ is minimal and $\tau$-infrequent, $|R_W| \le \tau$ and $|R_S| >
    \tau$ for all subsets $S \subset W$ such that $|S| = |W| - 1$, $S \ne
    \emptyset$. Let $W^\prime = W \setminus \{ w \} \cup \{ w^\prime \}$. We
    have $R_{W^\prime} = R_{W \setminus \{ w \}} \cap R_{w^\prime} =
    R_{W \setminus \{ w \}} \cap R_w = R_W$ since $R_w = R_{w^\prime}$. Hence,
    $|R_{W^\prime}| = |R_W| \le \tau$. Now consider any subset $S^\prime
    \subset W^\prime$ such that $|S^\prime| = |W^\prime| - 1$. We have
    $|W^\prime| - 1 = |W| - 1$ and either (i) $S^\prime = S$ when $w \notin S$
    or (ii) $S^\prime = S \setminus \{ w \} \cup \{ w^\prime \}$ when $w \in
    S$, where $S \subset W$, $|S| = |W| - 1$. Thus, either (i) $R_{S^\prime} =
    R_S$ or (ii) $R_{S^\prime} = R_{S \setminus \{ w \}} \cap R_{w^\prime} =
    R_{S \setminus \{ w \}} \cap R_w = R_S$, respectively. That is,
    $|R_{S^\prime}| = |R_S| > \tau$ and we are done. \qed
  \end{proof}

  It follows that the importance of the partitioning into $L_{A,\tau}$ and
  $I_{A,\tau}^\prime \setminus L_{A,\tau}$ is that after finding the set of
  $\tau$-infrequent and minimal itemsets $\mathcal{L}_{A,\tau} \subset 2^{L_{A,
  \tau}}$ of $L_{A,\tau}$, the set of $\tau$-infrequent and minimal itemsets
  $\mathcal{I}_{A,\tau} \subset 2^{I_A}$ of $I_A$ can be obtained immediately.
  Namely,
  \begin{proposition} \label{prop:reconstruct}
    For any partition $(L_{A,\tau}, I_{A,\tau}^\prime \setminus L_{A,\tau})$
    the following holds: $\mathcal{I}_{A,\tau} = \mathcal{L}_{A,\tau} \cup
    \mathcal{\bar{L}}_{A,\tau} \cup r_{A,\tau}$, where $\mathcal{\bar{L}}_{A,
    \tau} = \{  I \setminus \{ a \} \cup \{ b \} : I \in \mathcal{L}_{A,\tau},
    a \in I, b \in \bar{L}_{A,\tau}, R_a = R_b \}$.
  \end{proposition}
  \begin{proof}
    The proposition states that itemset $I \in \mathcal{I}_{A,\tau} \iff I \in
    \mathcal{L}_{A,\tau} \cup \mathcal{\bar{L}}_{A,\tau} \cup r_{A,\tau}$.
    ``$\Leftarrow$'' If itemset $I \in \mathcal{L}_{A,\tau}$ or $I \in r_{A,
    \tau}$ then $I$ is minimal and $\tau$-infrequent and so $I \in
    \mathcal{I}_{A,\tau}$; if $I \in \mathcal{\bar{L}}_{A,\tau}$ then, by
    Proposition \ref{prop:partition}, $I$ is minimal and $\tau$-infrequent and
    so $I \in \mathcal{I}_{A,\tau}$.
    \\
    ``$\Rightarrow$'' Suppose $I \in \mathcal{I}_{A,\tau}$. First of all
    observe that $\mathcal{\tilde{I}}_{A,\tau} = \mathcal{I}_{A,\tau}$, where
    $\tilde{I}_A = I_A \setminus U_A$ and $\mathcal{\tilde{I}}_{A,\tau}$ is the
    set of minimal and $\tau$-infrequent itemsets in $2^{\tilde{I}_A}$. This
    holds because $I \cap U_A = \emptyset$ for any $I \in \mathcal{I}_{A,\tau}$
    (suppose $u \in I$, $u \in U_A$ and $I$ is minimal and $\tau$-infrequent,
    then $R_I = R_{I \setminus \{ u \}} \cap R_u = R_{I \setminus \{ u \}}$
    since $R_u$ contains all rows of $A$; thus $|R_{I \setminus \{ u \}}| =
    |R_I| \le \tau$ which contradicts the minimality of $I$).
    Further, we have $\mathcal{\tilde{I}}_{A,\tau} = \mathcal{\hat{I}}_{A,\tau}
    \cup r_{A,\tau}$ where $\hat{I}_{A,\tau} = I_A \setminus U_A \setminus
    r_{A,\tau}$ and $\mathcal{\hat{I}}_{A,\tau}$ is the set of minimal and
    $\tau$-infrequent itemsets in $2^{\hat{I}_{A,\tau}}$. This is because the
    elements of $r_{A,\tau}$ are minimal and $\tau$-infrequent individual items
    and so if $I \in \mathcal{\tilde{I}}_{A,\tau}$ then either (i) $I \cap
    r_{A,\tau} = \emptyset$ or (ii) $|I| = 1$, $I \in r_{A,\tau}$ (if $|I \cap
    r_{A,\tau}| > 1$ then $|I| > 1$ and $|R_a| \le \tau$ $\forall a \in I \cap
    r_{A,\tau}$ and so $I$ is not minimal; if $|I \cap r_{A,\tau}| = 1$ and
    $|I| > 1$ then $I$ is not minimal).
    Hence, we have that $\mathcal{I}_{A,\tau} = \mathcal{\hat{I}}_{A,\tau} \cup
    r_{A,\tau}$. Now $\hat{I}_{A,\tau} = L_{A,\tau} \cup \bar{L}_{A,\tau}$ with
    $L_{A,\tau} \cap \bar{L}_{A,\tau} = \emptyset$. Hence, if $I \in
    \mathcal{\hat{I}}_{A,\tau}$ and $I \cap \bar{L}_{A,\tau} = \emptyset$ (so
    $I \subseteq L_{A,\tau}$) then $I \in \mathcal{L}_{A,\tau}$. If $I \in
    \mathcal{\hat{I}}_{A,\tau}$ and $I \cap \bar{L}_{A,\tau} \ne \emptyset$
    then $I \in \mathcal{\bar{L}}_{A,\tau}$ and we are done. Notice that this
    proof works for any partition $(L_{A,\tau}, I_{A,\tau}^\prime \setminus
    L_{A,\tau})$. \qed
  \end{proof}

  In light of Proposition \ref{prop:reconstruct}, our goal can therefore be
  simplified to finding all $\tau$-infrequent and minimal $k$-itemsets of
  $L_{A,\tau}$, $1 \le k \le k_{max}$.

  \begin{example} \label{ex2}
    For $\tau = 1$ and the dataset
    \begin{align*}
      A &= 
      \begin{bmatrix}
        1 & 2 & 3 & 4 & 8 \\
        1 & 2 & 7 & 4 & 8 \\
        1 & 6 & 3 & 4 & 8 \\
        5 & 2 & 3 & 4 & 9
      \end{bmatrix}
    \end{align*}
    we have
    \begin{align*}
      I_A &= \{ (1, 1, \{1, 2, 3\}), (2, 2, \{1, 2, 4\}), (3, 3, \{1, 3, 4\}),
      (4, 4, \{1, 2, 3, 4\}), \\
      & \hspace{0.65cm} (5, 1, \{4\}), (6, 2, \{3\}), (7, 3, \{2\}), (8, 5,
      \{1, 2, 3\}), (9, 5, \{4\}) \}. \\
      \delta_A &= \{ (5, 1, \{4\}), (6, 2, \{3\}), (7, 3, \{2\}), (9, 5, \{4\})
      \}. \\
      U_A &= \{ (4, 4, \{1, 2, 3, 4\}) \}. \\
      r_{A,\tau} &= \delta_A.
    \end{align*}
    The remaining set of non-uniform and non-unique items is
    $$
      I_{A,\tau}^\prime = I_A \setminus U_A \setminus r_{A,\tau} = \{ (1, 1,
      \{1, 2, 3\}), (2, 2, \{1, 2, 4\}), (3, 3, \{1, 3, 4\}), (8, 5, \{1, 2,
      3\}) \}.
    $$
    The set $I_{A,\tau}^\prime$ can be partitioned into sets $L_{A,\tau} = \{
    (1, 1, \{1, 2, 3\}), (2, 2, \{1, 2, 4\}), (3, 3,$ $\{1, 3, 4\}) \}$ and
    $\bar{L}_{A,\tau} = I_{A,\tau}^\prime \setminus L_{A,\tau} = \{ (8, 5,
    \{1, 2, 3\}) \}$ such that (i) $R_a \ne R_b$ $\forall a, b \in L_{A,\tau}$
    ($\{1, 2, 3\} \ne \{1, 2, 4\} \ne \{1, 3, 4\}$ and $\{1, 2, 3\} \ne \{1, 3,
    4\}$), (ii) $\forall c \in \bar{L}_{A,\tau}$ there exists $d \in L_{A,
    \tau}$ with $R_c = R_d$ (for $(8, 5, \{1, 2, 3\})$ there is $(1, 1, \{1, 2,
    3\})$ in $L_{A,\tau}$).

    Let $a = (1, 1, \{1, 2, 3\})$, $b = (2, 2, \{1, 2, 4\})$, $c = (3, 3, \{1,
    3, 4\})$ and $d = (8, 5,$ $\{1, 2, 3\})$. Proposition \ref{prop:partition}
    says that if $\{ a, b, c \} \subseteq L_{A,\tau}$ is a minimal
    $\tau$-infrequent itemset (which it is when $\tau = 1$) and $d \in I_A
    \setminus L_{A,\tau}$ with $R_d = R_a$ then $\{ d, b, c \}$ is also a
    minimal $\tau$-infrequent itemset. Proposition \ref{prop:reconstruct} says
    that for our chosen partition $(L_{A,\tau}, I_{A,\tau}^\prime \setminus
    L_{A,\tau})$ the set of all minimal $\tau$-infrequent itemsets
    $\mathcal{I}_{A,\tau}$ can be obtained from the sets $\mathcal{L}_{A,
    \tau}$, $\mathcal{\bar{L}}_{A,\tau}$ and $r_{A,\tau}$. \qed
  \end{example}

  \subsection{Pruning the Search Space} \label{sec:pruning}

  \begin{figure}
    \centering
    \includegraphics[width=0.55\columnwidth,trim=0cm 1.5cm 0cm 0cm]
    {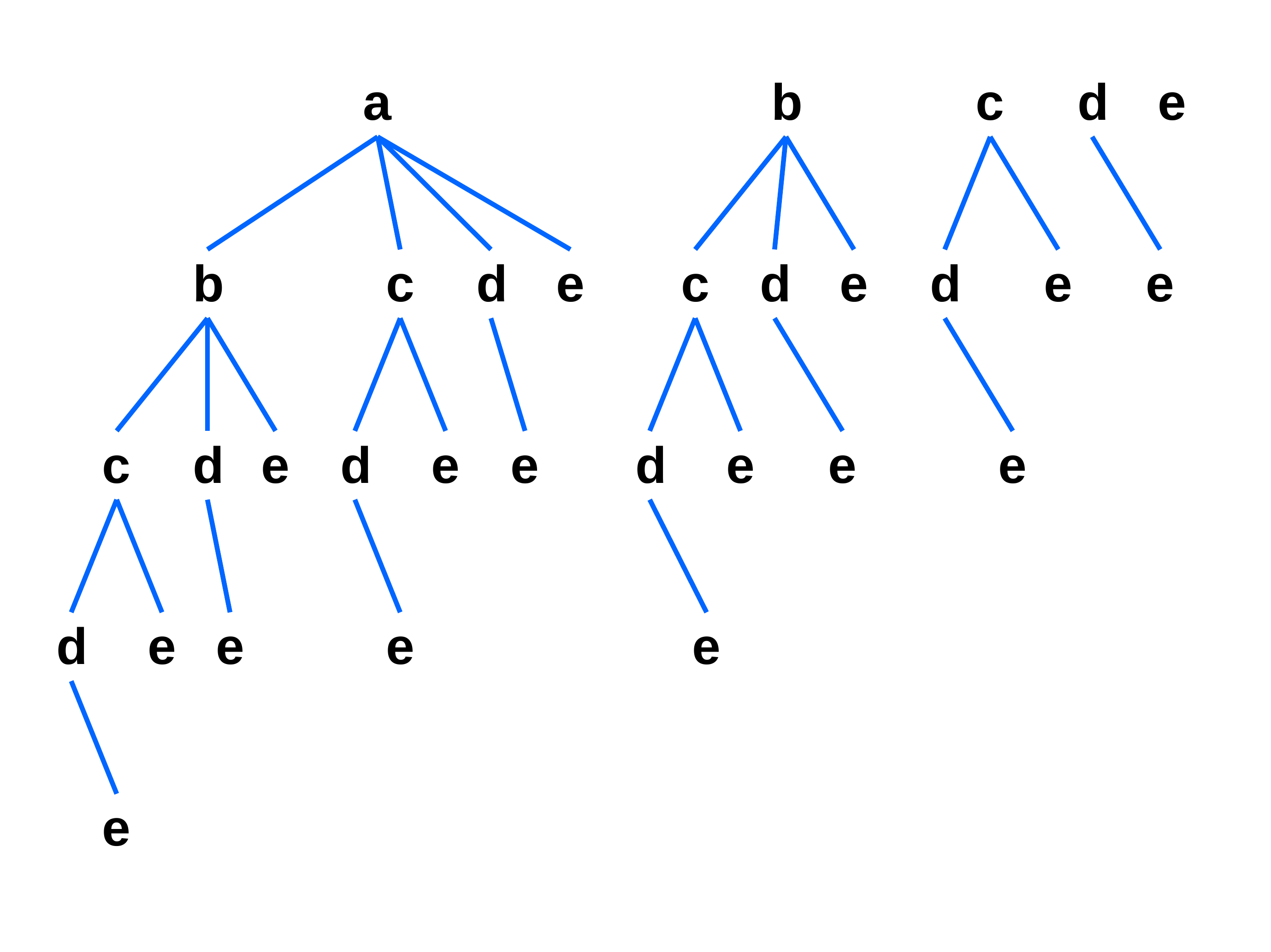}
    \caption{Prefix tree for the alphabet $L_{A,\tau} = \{ a, b, c, d, e \}$.
    By starting at the root and traversing the branches of the tree, every
    possible ordered sequence of letters can be obtained \emph{e.g.} traversing
    the far left-hand branch yields the sequence $abcde$.}
    \label{fig1}
  \end{figure}

  Considering the items in $L_{A,\tau}$ to be an alphabet, all of the possible
  words in the form of ordered sequences that can be built from $L_{A,\tau}$
  can be represented by a prefix tree. For example, when $L_{A,\tau} = \{ a, b,
  c, d, e \}$, the associated prefix tree is shown in the Figure \ref{fig1}. By
  starting at the root and traversing the branches of the tree, every possible
  ordered sequence of letters can be obtained.

  In principle, the $\tau$-infrequent and minimal $k$-itemsets of $L_{A,\tau}$
  can be found by traversing every branch of the tree to depth $k_{max}$ and
  testing each sequence of items obtained for $\tau$-infrequency and
  minimality. However, efficiency can be increased if it is possible to avoid
  fully traversing every branch i.e. the tree can be pruned.
  Basic pruning can be achieved using following fundamental property of
  itemsets:
  \begin{proposition}[Monotonicity] \label{prop:mono}
    Let $I$ be an itemset. If $I$ is not minimal then no superset of $I$ can be
    minimal.
  \end{proposition}
  \begin{proof}
    Since $I$ is non-minimal there exists $S \subset I$, $S \ne \emptyset$ such
    that $|R_S| \le \tau$. It follows that $\forall J \supset I$ there exists
    $S \subset J$, $S \ne \emptyset$ such that $|R_S| \le \tau$ and so $J$ is
    also non-minimal. \qed
  \end{proof}

  Hence, as soon as we determine that the sequence of items in an itemset is
  non-minimal, we can terminate traversal of that branch of the tree. Note that
  similar pruning is not possible based on $\tau$-infrequency since a superset
  of an itemset $I$ can be $\tau$-infrequent even if $I$ is not
  $\tau$-infrequent due to the decrease in frequency as more and more items are
  added to an itemset.

  Importantly, the prefix tree associated with itemset $L_{A,\tau}$ is not
  unique since the tree depends on how we choose to order the items in $L_{A,
  \tau}$. In general, it is challenging to determine an ordering of items in
  $L_{A,\tau}$ which minimises the number of vertices which need to be
  traversed in the prefix tree in order to find the set $\mathcal{L}_{A,\tau}$
  of $\tau$-infrequent and minimal itemsets of $L_{A,\tau}$. We revisit this
  question later, in Section \ref{sec:pruningperf}, but note here that sorting
  the items of $L_{A,\tau}$ into ascending order using the following item
  ordering is efficient for a wide range of datasets.

  \begin{definition}[Ascending Order] \label{def:order}
    We order items $a < b$ if (i) $|R_a| < |R_b|$ or (ii) $|R_a| = |R_b|$ and
    $j_a < j_b$ or (iii) $|R_a| = |R_b|$, $j_a = j_b$ and $\min R_a < \min
    R_b$.
  \end{definition}

  Note that due to the pre-processing and partitioning used to obtain $L_{A,
  \tau}$, for any items $a \in L_{A,\tau}$, $b\in L_{A,\tau} \setminus \{ a \}$
  we must have either $a < b$ or $b < a$ \emph{i.e.} strict total order (if
  $j_a = j_b$, $\min R_a = \min R_b$ then items $a$ and $b$ are both in the
  same column $j_a$ and row $\min R_a$ of the dataset and so we must have $a =
  b$, but this contradicts the fact that $b \in L_{A,\tau} \setminus \{ a \}$).
  We let $L_{A,\tau}^<$ denote a list of the items in $L_{A,\tau}$ sorted in
  ascending order. Note that $L_{A,\tau}^<$ is simply a permutation of $L_{A,
  \tau}$.

  \subsection{Potential Performance Bottlenecks}

  To evaluate whether an itemset $I$ is minimal or not we use the support
  itemset test to verify Definition \ref{def:mri}(2). To evaluate whether an
  itemset $I$ is $\tau$-infrequent, we intersect the rows of the elements in
  $I$ to obtain $R_I = \cap_{a \in I} R_a$ and test whether $|R_I| \le \tau$ to
  verify Definition \ref{def:mri}(1). Both of these tests are potentially
  expensive.

  The support itemset test requires enumerating the subsets $S\subset I$,
  $|S| = |I| - 1$, and calculating $R_S = \cap_{a \in S} R_a$ for each subset.
  As already noted, testing for $\tau$-infrequency requires calculating $R_I =
  \cap_{a \in I} R_a$. For large tables, the row sets $R_a$ may be large and so
  time consuming to obtain, \emph{e.g.} if the approach taken is to scan the
  dataset for item $a$ and record the rows in which $a$ appears, plus
  additionally the complexity of calculating $R_I$ in the obvious manner scales
  as $O(|I|\min_{a\in I}|R_a|)$.

  \subsection{Kyiv Algorithm}

  The Kyiv algorithm performs a breadth first search of the prefix tree defined
  by ordered list $L_{A,\tau}^<$. Branches are pruned using Proposition
  \ref{prop:mono} -- if an itemset $I$ fails the support itemset test in
  Definition \ref{def:mri}(2) then it must be non-minimal and so the subtree
  with itemset $I$ at the root can be pruned. The key advantage of the
  breadth-first approach is that the support row test can be performed
  extremely efficiently, as discussed in more detail in Section
  \ref{sec:supitest}. Pseudo-code for the Kyiv algorithm is given in Algorithm
  \ref{al:kyiv}.

  \begin{algorithm}\caption{Kyiv} \label{al:kyiv}
  \begin{algorithmic}[1]
    \State \textbf{Input}: dataset $A$, $\tau$, threshold $k_{max}$
    \State \textbf{Output}: all minimal $\tau$-infrequent $k$-itemsets, $k \le
    k_{max}$
    \State compute $I_A = I_{A,\tau}^\prime \cup U_A \cup r_{A,\tau}$
    \State compute $L_{A,\tau}$ for chosen partition $(L_{A,\tau},
    I_{A,\tau}^\prime \setminus L_{A,\tau})$
    \State print $\tau$-infrequent items in $r_{A,\tau}$ \Comment{$k = 1$ case}
    \State sort $L_{A,\tau}$ to obtain $L_{A,\tau}^<$
    \State $t \leftarrow 0$, $k \gets 2$
    \ForEach {$a \in L_{A,\tau}^<$} $t \leftarrow t + 1$, $P_t \leftarrow
    \{ a \}$
    \label{alg:prep}
    \EndFor
    \While {$k \le k_{max}$} \label{alg:while}
      \State $t^\prime \gets 0$
      \ForEach {$i \in \{ 1, \dots, t - 1 \}$}
        \State $I \gets P_i$
        \ForEach {$j \in \{ i + 1, \dots, t \}$} \label{alg:thread1}
          \State $J \gets P_j$
          \LineComment {get the highest order items in $I$ and $J$}
          \State $a \leftarrow \max(I), \ b \leftarrow \max(J)$
          \If {$I \setminus \{ a \} \ne J \setminus \{ b \}$}
            \State break \Comment {itemsets do not share a common prefix}
          \EndIf
          \LineComment {itemsets $I$ and $J$ differ exactly by one item now}
          \State $W \gets I \cup J$
          \If {$k > 2$} \label{alg:kgr2condition}
            \LineComment {support itemset test, Definition \ref{def:mri}(2)}
            \If {$\exists \, S \subset W, |S| = |W| - 1: |R_S| \le \tau$}
            \label{alg:support}
              \State continue
              \Comment {non-minimal, prune this branch}
            \EndIf
            \If {$k = k_{max}$}
              \LineComment {Lemma \ref{lem} and Corollary \ref{cor}}
              \If {$|R_I| + |R_J| > |R_{I \setminus \{ a \}}| + \tau$}
              continue
              \label{alg:lem}
              \EndIf
              \State $c \gets \max (J \setminus \{ b \})$
              \NoThen
              \If {$\min(|R_{I \setminus \{ c \}}| - |R_I|, |R_{J \setminus \{
              c \}}| - |R_J|) + \tau < |R_{I \setminus \{ c \}} \cap R_b|$}
              \label{alg:cor} \textbf{then}
                \State continue
              \EndIf
              \ReThen
            \EndIf
          \EndIf
          \State $R_W \gets R_I \cap R_J$ \Comment{intersect rows}
          \label{alg:intersect}
          \If {$|R_W| = 0$ or $|R_W| = \min(|R_I|, |R_J|)$}
            \State continue
            \Comment {skip absent and uniform itemsets}
          \EndIf
          \If {$|R_W| \le \tau$}
            \State print $W$
            \Comment {minimal $\tau$-infrequent itemset found}
            \ForEach {$w \in W$}
            \Comment {apply Proposition \ref{prop:partition}}
              \If {$\exists \, w' \in I_{A,\tau}^\prime \setminus L_{A,\tau}:
              R_w = R_{w'}$}
                \State print $W \setminus \{ w \} \cup \{ w' \}$
              \EndIf
            \EndFor
          \Else \Comment {need to store non-$\tau$-infrequent minimal itemset}
            \If {$k < k_{max}$}
               \State $t^\prime \gets t^\prime + 1$, $P_{t^\prime}^\prime \gets
               W$
            \EndIf
          \EndIf
        \EndFor
      \EndFor
      \ForEach {$t \in \{ 1, \dots, t^\prime \}$} \label{alg:copystructure}
        $P_t \gets P_t^\prime$
      \EndFor
      \State $k \gets k + 1$, $t \gets t^\prime$
    \EndWhile
  \end{algorithmic}
  \end{algorithm}

  In Algorithm \ref{al:kyiv} the collection of sets $\{ P_i \}_{i = 1}^t$ holds
  the vertices of level $k - 1$ of the pruned prefix graph, and the vertices of
  level $k$ are stored in $\{ P_i^\prime \}_{i = 1}^{ t^\prime }$. Note that
  there is never any need to store more than two levels of the pruned prefix
  tree -- we discuss these memory requirements in more detail below. The
  algorithm visits each vertex in level $k$ and takes one of three actions: (i)
  finds that the vertex is a non-minimal itemset and so prunes it (it is not
  added to $P^\prime$ and its children are not traversed), (ii) finds that the
  vertex is a minimal $\tau$-infrequent itemset and so prints it (it is not
  added to $P^\prime$ and its children are not traversed), (iii) finds that the
  vertex is not $\tau$-infrequent and its children must be traversed.

  In our implementation of Algorithm \ref{al:kyiv}, we use a recursive data
  structure called Graph to hold the prefix tree levels. Graph stores an array
  of references to its children of type Graph and other useful data such as the
  rows associated with the current node. Each child is an item $(v, j_a, R_a)$
  and is identified by index value $m v + j_a$. Fast access to the children is
  achieved by use of a hash table, which is also stored among the properties of
  the Graph class.

  \subsubsection{Highly Efficient Support Itemset Testing} \label{sec:supitest}

  One of the key benefits of adopting a breadth-first approach in Algorithm
  \ref{al:kyiv} is that the computational cost of the support itemset test at
  line \ref{alg:support} can be reduced to essentially zero. This is because
  the itemsets $S \subset W$ of size $|S| = |W| - 1$, together with the
  associated row sets $R_S$, have already been pre-calculated and stored in
  data structure $\mathcal{P} := \{ P_i \}_{i = 1}^t$. Hence, evaluating
  whether there exists an $S$ such that $|R_S| \le \tau$ simply involves
  lookups from $\mathcal{P}$, which can be carried out efficiently using an
  appropriate data structure for $\mathcal{P}$.

  Observe that acceleration of the support itemset test at line
  \ref{alg:support} is achieved in Algorithm \ref{al:kyiv} at the cost of
  increased memory usage to store data structure $\mathcal{P}$. As $\tau$
  increases, the number of prefix tree vertices decreases and the arrays stored
  at each vertex occupy less memory. Nevertheless, this memory cost remains
  potentially significant, particularly when $\tau$ is small and in the middle
  of the prefix tree where the number of vertices in each level of the tree is
  largest. However, in view of the fact that the amount of RAM available is
  growing at a much faster rate than CPU clock speed, this trade-off between of
  increased memory consumption for a much reduced computational burden can be a
  favourable one.

  \subsubsection{Reducing Number of Row Intersections}
  \label{sec:lemcor}

  The remaining computational bottleneck of Algorithm \ref{al:kyiv} is at line
  \ref{alg:intersect}. We present performance measurements in Section
  \ref{sec:exp} that confirm line \ref{alg:intersect} accounts for the vast
  majority of the execution time of Algorithm \ref{al:kyiv}. However, we leave
  as future work the development of more efficient techniques for computing the
  intersection operation at line \ref{alg:intersect}.

  The potential exists to reduce the number of row intersections at the
  $k_{max}$ level of the prefix tree using the following properties:

  \begin{lemma} \label{lem}
    Let $I \subseteq I_A$ be an itemset and $a, b \in I_A$ any items in $I_A$.
    If
    \begin{align}
      |R_I \cap R_a| + |R_I \cap R_b| > |R_I| + \tau
    \end{align}
    then $I \cup \{ a, b \}$ is not a $\tau$-infrequent itemset.
  \end{lemma}
  \begin{proof}
    We proceed by contradiction. Suppose $|R_I \cap R_a| + |R_I \cap R_b| >
    |R_I| + \tau$ and itemset $I \cup \{ a, b \}$ is $\tau$-infrequent (so
    $|R_I \cap R_a \cap R_b| \le \tau$). By the distributivity of set
    intersection, $R_I \cap (R_a \cup R_b) = (R_I \cap R_a) \cup (R_I \cap
    R_b)$. Hence,
    \begin{align*}
      & |R_I \cap (R_a \cup R_b)| \\
      & \quad = |(R_I \cap R_a) \cup (R_I \cap R_b)| \\
      & \quad = |R_I \cap R_a| + |R_I \cap R_b| - |(R_I \cap R_a) \cap (R_I
      \cap R_b)| \\
      & \quad = |R_I \cap R_a| + |R_I \cap R_b| - |R_I \cap R_a \cap R_b|.
    \end{align*}
    Now $|R_I| \ge |R_I \cap (R_a \cup R_b)|$ and by assumption $|R_I \cap R_a
    \cap R_b| \le \tau$. Hence, $|R_I| \ge |R_I \cap R_a| + |R_I \cap R_b|
    - \tau$, yielding the desired contradiction. \qed
  \end{proof}

  \begin{corollary} \label{cor}
    Let $a_1, \dots, a_k \in I_A$ be any items from $I_A$, with $k > 2$. If
    \begin{align}
      & \Gamma_0 > \min \{ \Gamma_1, \Gamma_2 \} + \tau
    \end{align}
    then $\{ a_1, \dots, a_k \}$ is not a $\tau$-infrequent itemset, where
    \begin{align*}
      \Gamma_0 & := |\cap_{i = 1}^{k-3} R_{a_i} \cap R_{a_{k-1}} \cap R_{a_k}|,
      \\
      \Gamma_1 & := |\cap_{i = 1}^{k-3} R_{a_i} \cap R_{a_{k-1}}| -
      |\cap_{i = 1}^{k-3} R_{a_i} \cap R_{a_{k-2}} \cap R_{a_{k-1}}|, \\
      \Gamma_2 & := |\cap_{i = 1}^{k-3} R_{a_i} \cap R_{a_k}| -
      |\cap_{i = 1}^{k-3} R_{a_i} \cap R_{a_{k-2}} \cap R_{a_k}|.
    \end{align*}
  \end{corollary}
  \begin{proof}
    There are two cases to consider.
    \newline
    Case (i): $\Gamma_0 > \min \{ \Gamma_1, \Gamma_2 \} + \tau = \Gamma_1 +
    \tau$. Then,
    \begin{align*}
      & |\cap_{i = 1}^{k-3} R_{a_i} \cap R_{a_{k-1}} \cap R_{a_{k-2}}| +
      |\cap_{i = 1}^{k-3} R_{a_i} \cap R_{a_{k-1}} \cap R_{a_k}| \\
      & \qquad > |\cap_{i = 1}^{k-3} R_{a_i} \cap R_{a_{k-1}}| + \tau.
    \end{align*}
    Let $I = \cup_{i = 1}^{k-3} a_i \cup \{ a_{k-1} \}$, $a = a_{k-2}$, $b =
    a_k$. By Lemma \ref{lem} $\{ a_1, \dots, a_k \}$ is not $\tau$-infrequent.
    \newline
    Case (ii): $\Gamma_0 > \min \{ \Gamma_1, \Gamma_2 \} + \tau = \Gamma_2 +
    \tau$. Then,
    \begin{align*}
      & |\cap_{i = 1}^{k-3} R_{a_i} \cap R_{a_k} \cap R_{a_{k-2}}| +
      |\cap_{i = 1}^{k-3} R_{a_i} \cap R_{a_k} \cap R_{a_{k-1}}| \\
      & \qquad > |\cap_{i = 1}^{k-3} R_{a_i} \cap R_{a_k}| + \tau.
    \end{align*}
    Let $I = \cup_{i = 1}^{k-3} a_i \cup \{ a_k \}$, $a = a_{k-2}$, $b =
    a_{k - 1}$. By Lemma \ref{lem} $\{ a_1, \dots, a_k \}$ is not
    $\tau$-infrequent. \qed
  \end{proof}

  In the final iteration (when $k = k_{max}$) we can use Lemma \ref{lem} and
  Corollary \ref{cor} to test for $\tau$-infrequency before carrying out the
  intersection at line \ref{alg:intersect}. If either test concludes that the
  itemset is not $\tau$-infrequent, then there is no need to perform the row
  intersection.

  \begin{example} \label{ex3}
    To illustrate the operation of Algorithm \ref{al:kyiv}, suppose $k_{max} =
    3$, $\tau = 1$ and consider the dataset:
    \begin{align*}
      A = \
      \begin{bmatrix}
        * & * & * & 4 & * \\
        1 & 2 & * & 4 & * \\
        1 & 2 & 3 & 4 & * \\
        1 & 2 & 3 & 4 & 5 \\
        1 & * & 3 & * & 5 \\
        * & 2 & 3 & * & 5 \\
        * & * & * & * & 5
      \end{bmatrix}
      \text{, where $*$ denotes a unique item.}
    \end{align*}
    The set $r_{A,\tau}$ contains the unique items marked by $*$. There are no
    uniform items, so $U_A = \emptyset$. There exists single partition of
    $I_{A,\tau}^\prime$ -- $(L_{A,\tau}, \emptyset)$, where it can be verified
    that \small
    \begin{align*}
      L_{A,\tau}^< &= \{ (1, 1, \{2, 3, 4, 5\}), (2, 2, \{2, 3, 4, 6\}), (3, 3,
      \{3, 4, 5, 6\}),
      (4, 4, \{1, 2, 3, 4\}), (5, 5, \{4, 5, 6, 7\}) \} \\
      &:= \{ a, b, c, d,
      e \}.
    \end{align*}
    \normalsize
    The prefix tree of $L_{A,\tau}^<$ is shown schematically in Figure
    \ref{fig1}. After line \ref{alg:prep} is executed ($P_1 = \{ a \}, P_2 =
    \{ b \}, P_3 = \{ c \}, P_4 = \{ d \}, P_5 = \{ e \}$) and the first level
    of the prefix tree is built. The first iteration of the main loop at line
    \ref{alg:while} (when $k = 2 < 3 = k_{max}$ and $t = 5$) is reproduced
    step-by-step below. Here, $1 \le i \le 4 = t - 1, i < j \le t$ and for each
    $(I, J)$ the highest order items are the items contained in $I$ and $J$
    (which never share a common prefix). The condition at line
    \ref{alg:kgr2condition} is false and there are no absent or uniform
    itemsets ($0 < |R_W| < \min(|R_I|, |R_J|)$ for each $(I, J)$) after
    intersection at line \ref{alg:intersect}:
    \small
    \begin{align*}
      I = P_1 = \{ a \}:\quad
      & J = P_2 = \{ b \},\ W = \{ a, b \},\ R_W = \{2, 3, 4\} \Rightarrow
      P_1^\prime = \{ a, b \} \\
      & J = P_3 = \{ c \},\ W = \{ a, c \},\ R_W = \{3, 4, 5\} \Rightarrow
      P_2^\prime = \{ a, c \} \\
      & J = P_4 = \{ d \},\ W = \{ a, d \},\ R_W = \{2, 3, 4\} \Rightarrow
      P_3^\prime = \{ a, d \} \\
      & J = P_5 = \{ e \},\ W = \{ a, e \},\ R_W = \{4, 5\} \Rightarrow
      P_4^\prime = \{ a, e \} \\
      I = P_2 = \{ b \} :\quad
      & J = P_3 = \{ c \},\ W = \{ b, c \},\ R_W = \{3, 4, 6\} \Rightarrow
      P_5^\prime = \{ b, c \} \\
      & J = P_4 = \{ d \},\ W = \{ b, d \},\ R_W = \{2, 3, 4\} \Rightarrow
      P_6^\prime = \{ b, d \} \\
      & J = P_5 = \{ e \},\ W = \{ b, e \},\ R_W = \{4, 6\} \Rightarrow
      P_7^\prime = \{ b, e \} \\
      I = P_3 = \{ c \} :\quad
      & J = P_4 = \{ d \},\ W = \{ c, d \},\ R_W = \{3, 4\} \Rightarrow
      P_8^\prime = \{ c, d \} \\
      & J = P_5 = \{ e \},\ W = \{ c, e \},\ R_W = \{4, 5, 6\} \Rightarrow
      P_9^\prime = \{ c, e \} \\
      I = P_4 = \{ d \}:\quad &J = P_5 = \{ e \},\ W = \{ d, e \},\ R_W = \{4\}
      \Rightarrow \text{print $\{ d, e \}$}
    \end{align*}
    \normalsize
    The second level of the prefix tree is now built: $P_1 = \{ a, b \}, P_2 =
    \{ a, c \}, P_3 = \{ a, d \},$ $P_4 = \{ a, e \}, P_5 = \{ b, c \}, P_6 =
    \{ b, d \}, P_7 = \{ b, e \}, P_8 = \{ c, d \}, P_9 = \{ c, e \}$.

    The second iteration of the main loop (when $k = 3 = k_{max}$ and $t = 9$)
    is reproduced step-by-step below. Here, $1 \le i \le 8 = t - 1, i < j \le
    t$:
    \small
    \begin{align*}
      I = P_1 = \{ a, b \}:\quad
      & J = P_2 = \{ a, c \},\ W = \{ a, b, c \} \\
      & J = P_3 = \{ a, d \},\ W = \{ a, b, d \} \\
      & J = P_4 = \{ a, e \},\ W = \{ a, b, e \},\ 
      R_W = \{4\} \Rightarrow \text{print $\{ a, b, e \}$} \\
      I = P_2 = \{ a, c \} :\quad
      & J = P_3 = \{ a, d \},\ W = \{ a, c, d \} \\
      & J = P_4 = \{ a, e \},\ W = \{ a, c, e \} \\
      I = P_3 = \{ a, d \}:\quad &J = P_4 = \{ a, e \},\ W = \{ a, d, e \} \\
      I = P_5 = \{ b, c \}:\quad
      & J = P_6 = \{ b, d \},\ W = \{ b, c, d \} \\
      & J = P_7 = \{ b, e \},\ W = \{ b, c, e \} \\
      I = P_6 = \{ b, d \}:\quad &J = P_7 = \{ b, e \},\ W = \{ b, d, e \} \\
      I = P_8 = \{ c, d \}:\quad &J = P_9 = \{ c, e \},\ W = \{ c, d, e \}
    \end{align*}

    At the ultimate level $k_{max}$, the support itemset test for minimality
    (line \ref{alg:support}), Lemma \ref{lem} (line \ref{alg:lem}) and
    Corollary \ref{cor} (line \ref{alg:cor}) are applied in that order to pairs
    of $2$-itemsets from $P$ which share a common prefix. Pairs $(\{ a, d \}$,
    $\{ a, e \})$, $(\{ b, d \}$, $\{ b, e \})$, $(\{ c, d \}$, $\{ c, e \})$
    are pruned by the support itemset test. Pairs $(\{ a, b \}$, $\{ a, c \})$,
    $(\{ a, b \}$, $\{ a, d \})$, $(\{ a, c \}$, $\{ a, d \})$, $(\{ b, c \}$,
    $\{ b, d \})$ are pruned by the lemma. Pairs $(\{ a, c \}$, $\{ a, e \})$,
    $(\{ b, c \}$, $\{ b, e \})$ are pruned by the corollary. Leaving only
    $(\{ a, b \}$, $\{ a, e \})$ as minimal unique itemset. \qed
  \end{example}

  \subsubsection{Correctness}

  \begin{theorem}
    Algorithm \ref{al:kyiv} terminates in finite time and finds all minimal
    $\tau$-infrequent itemsets of $I_A$ up to size $k_{max}$.
  \end{theorem}
  \begin{proof}
    Pre-processing from the beginning to the main loop (line \ref{alg:while})
    is done in finite time: to compute $I_A$ and $L_{A,\tau}$ algorithm goes
    through the $A$ elements and counts their frequencies while the size of $A$
    is finite ($n, m < +\infty$); printing $r_{A,\tau}$, sorting $L_{A,\tau}$
    and iterating $|L_{A,\tau}^<|$ times the loop at line \ref{alg:prep} all
    take finite time as $|r_{A,\tau}|, |L_{A,\tau}| = |L_{A,\tau}^<| <
    +\infty$. The search space of the algorithm is the prefix tree which is
    finite as $I_A$ is finite. If there is no pruning then Algorithm
    \ref{al:kyiv} goes through every branch of maximum length $k_{max}$ of the
    tree, otherwise it processes even less number of branches. It takes finite
    time to process a single branch, that is: navigate it, intersect itemset
    rows of finite size and either print (Proposition \ref{prop:partition}
    takes finite time because $|W|, |I_{A,\tau}^\prime \setminus L_{A,\tau}| <
    +\infty$) or store the appropriate itemset. Consequently the algorithm
    terminates in finite time processing all the itemsets of maximum size
    $k_{max}$ that have not been thrown out by the support itemset test (line
    \ref{alg:support}), Lemma \ref{lem} (line \ref{alg:lem}) and Corollary
    \ref{lem} (line \ref{alg:cor}).

    Suppose there is a minimal $\tau$-infrequent itemset $I \in 2^{I_A}$ that
    is not found by the algorithm. Proposition \ref{prop:reconstruct} means
    that the set of all $\tau$-infrequent and minimal itemsets $\mathcal{I}_{A,
    \tau} \subset 2^{I_A}$ can be described by any chosen partition $(L_{A,
    \tau}, \bar{L}_{A,\tau})$. Thus, either $I$ contains item which does not
    belong to $L_{A,\tau}$ or $|I| > k_{max}$. The former is impossible while
    the latter does not contradict the theorem. \qed
  \end{proof}

  \subsubsection{Parallelisation}

  Algorithm \ref{al:kyiv} can be readily parallelised using shared-memory
  threads. Namely, at level $k$ within the prefix tree assign all vertices
  sharing the same parent at level $k - 1$ within the prefix tree to the same
  thread and then in each thread execute the loop starting at line
  \ref{alg:thread1} in Algorithm \ref{al:kyiv}. The shared memory allows each
  thread access to the prefix tree information stored in $P_j$, $j \in \{ i +
  1, \cdots, t \}$, but there is otherwise no need for inter-thread
  communication.

  When the number of available threads is less than the number of parent
  vertices at level $k - 1$ in the prefix tree, work must be allocated amongst
  the threads. As already discussed, the work associated with each parent
  vertex is dominated by the number of row intersections to be carried out.
  This number can be accurately estimated based on the number of children of
  the parent vertex, and so the work associated with each parent vertex
  estimated in advance. Using these work estimates, load-balanced scheduling of
  work amongst the threads can then be efficiently realised. As discussed in
  more detail in Section \ref{sec:exp}, in this way we can ensure that the
  running time of all threads is similar thereby enhancing the performance gain
  from parallelisation -- we note that imbalanced thread run times is known to
  be a key bottleneck in the parallelisation of state-of-the-art depth-first
  approaches such as SUDA2 and MINIT \cite{haglin09}.

  \begin{example}
    Recall Example \ref{ex3}. Let $t = 3$ be the number of threads. When $k =
    2$, Algorithm \ref{al:kyiv} allocates jobs between the $3$ threads: first
    an empty array $T$ of size $t$ is created; then for each item in $L_{A,
    \tau}^<$ the number of higher order items is stored in $T$ at the cell
    which has the minimum value (if there are several such cells, the left-most
    is chosen). As soon as $T$ is filled in, all threads start work. In our
    example $T = \{ 4, 3, 3 \}$ and the first thread is assigned itemsets,
    $\{ a, b \}$, $\{ a, c \}$, $\{ a, d \}$, $\{ a, e \}$, the second $\{ b,
    c \}$, $\{ b, d \}$, $\{ b, e \}$ and the third $\{ c, d \}$, $\{ c, e \}$,
    $\{ d, e \}$. Row intersection of each ordered pair reveals the unique
    $2$-itemsets and these itemsets are stored in $P^\prime$: $\{ a, b \}$,
    $\{ a, c \}$, $\{ a, d \}$, $\{ a, e \}$, $\{ b, c \}$, $\{ b, d \}$,
    $\{ b, e \}$, $\{ c, d \}$ and $\{ c, e \}$; at the next iteration they
    will be copied into $P$ for the $k = 3$ analysis. Only $\{ d, e \}$ will be
    printed out as unique and minimal.

    When $k = 3$ (the ultimate level $k_{max}$), $T = \{ 6, 3, 1 \}$ and the
    first thread is assigned itemsets $(\{ a, b \}$, $\{ a, c \})$, $(\{ a,
    b \}$, $\{ a, d \})$, $(\{ a, b \}$, $\{ a, e \})$, $(\{ a, c \}$, $\{ a,
    d \})$, $(\{ a, c \}$, $\{ a, e \})$, $(\{ a, d \}$, $\{ a, e \})$, the
    second $(\{ b, c \}$, $\{ b, d \})$, $(\{ b, c \}$, $\{ b, e \})$, $(\{ b,
    d \}$, $\{ b, e \})$ and the third $(\{ c, d \}$, $\{ c, e \})$. As in
    Example \ref{ex3}, the support itemset test, Lemma \ref{lem} and Corollary
    \ref{cor} eliminates all pairs inside the threads except for $(\{ a, b \}$,
    $\{ a, e \})$. \qed
  \end{example}

\section{Experiments} \label{sec:exp}

  Unless otherwise stated, all experiments in this section were carried out
  using ascending itemlist order, Lemma \ref{lem} and Corollary \ref{cor}.

  \subsection{Hardware and Software Setup} \label{sec:hardsoft}

  We implemented Algorithm \ref{al:kyiv} in Java (version 1.7.0\_25) using the
  hppc (version 0.5.2) library, which can be found at
  \url{http://labs.carrotsearch.com/hppc.html}. For comparison with the serial
  version of Algorithm \ref{al:kyiv}, we also implemented a state-of-the-art
  algorithm MINIT \cite{minit} in Java (using the C++ implementation kindly
  provided by the developers of MINIT) and used the C++ implementation of the
  MIWI algorithm \cite{cagliero13}, kindly provided by its developers.

  For testing we used an Amazon cr1.8xlarge instance with an Intel Xeon CPU
  E5-2670 $0$ @ $2.60$GHz $32$ processor (up to $32$ hyperthreads), $244$Gb of
  memory, $64$-bit Linux operating system (kernel version 3.4.62-53.42.
  amzn1.x86\_64 of Red Hat 4.6.3-2 Linux distribution (Amazon Linux AMI release
  2013.09)).

  \subsection{Domain-Agnostic Performance} \label{sec:domagnperf}

  \subsubsection{Randomised Datasets}

  We begin by investigating performance in a domain-agnostic manner using
  randomised datasets. Each randomised dataset consists of $50,000$ rows with
  each row having $25$ columns. For each column, the size $D$ of the domain of
  element values is selected i.i.d. uniformly at random from the set $\{ 10,
  \cdots, 100 \}$. The elements within each column are then selected i.i.d.
  uniformly at random from domain $\{ 1, \cdots, D \}$. On average, for these
  datasets $L_A$ contained $1352$ items.

  \subsubsection{Execution Time}

  \begin{figure}
    \centering
    \includegraphics[width=0.5\textwidth]{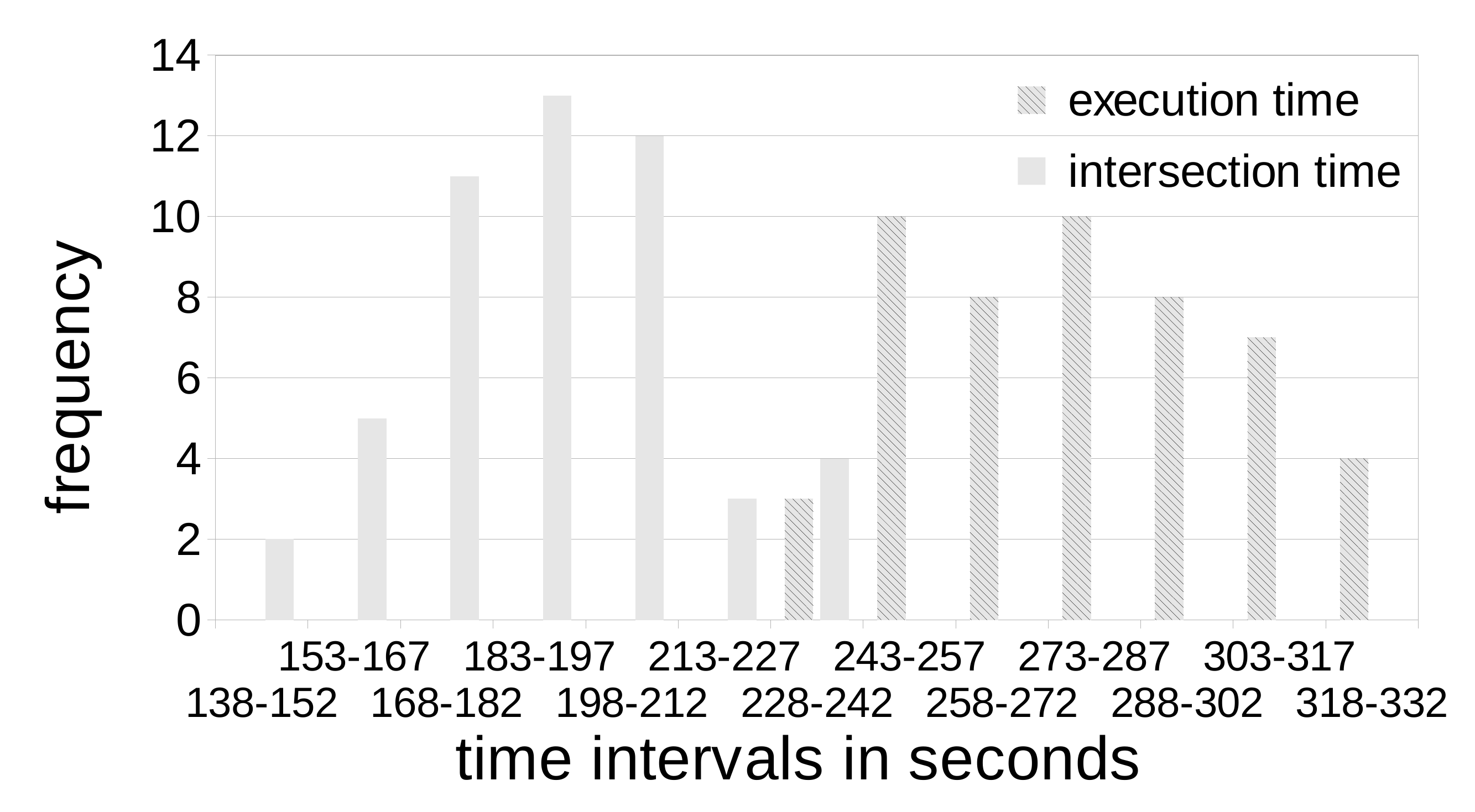}
    \captionof{figure}{Distribution of execution and intersection time for
    randomised datasets, $k_{max} = 5$, $\tau = 1$.}
    \label{fig2}
  \end{figure}

  Figure \ref{fig2} shows the measured distribution of execution times for
  Algorithm \ref{al:kyiv} over $50$ randomised datasets when $k_{max} = 5$,
  $\tau = 1$. It can be seen that the execution times are relatively tightly
  bunched around the mean value of $280$ seconds. Also shown in Figure
  \ref{fig2} is the corresponding time expended on calculating row
  intersections at line \ref{alg:intersect} of Algorithm \ref{al:kyiv}. The
  mean intersection time is $190$ seconds, so $68\%$ of the execution time is
  expended on row intersections, confirming that these are indeed the primary
  bottleneck in Algorithm \ref{al:kyiv}. Note that the fraction of execution
  time expended on row intersections depends on $k_{max}$ and tends to increase
  as $k_{max}$ decreases \emph{e.g.} when $k_{max} = 3$ row intersections
  absorb $80\%$ of the execution time.

  \subsubsection{Prefix Tree Pruning}

  \begin{figure}
    \centering
    \includegraphics[width=0.5\textwidth]{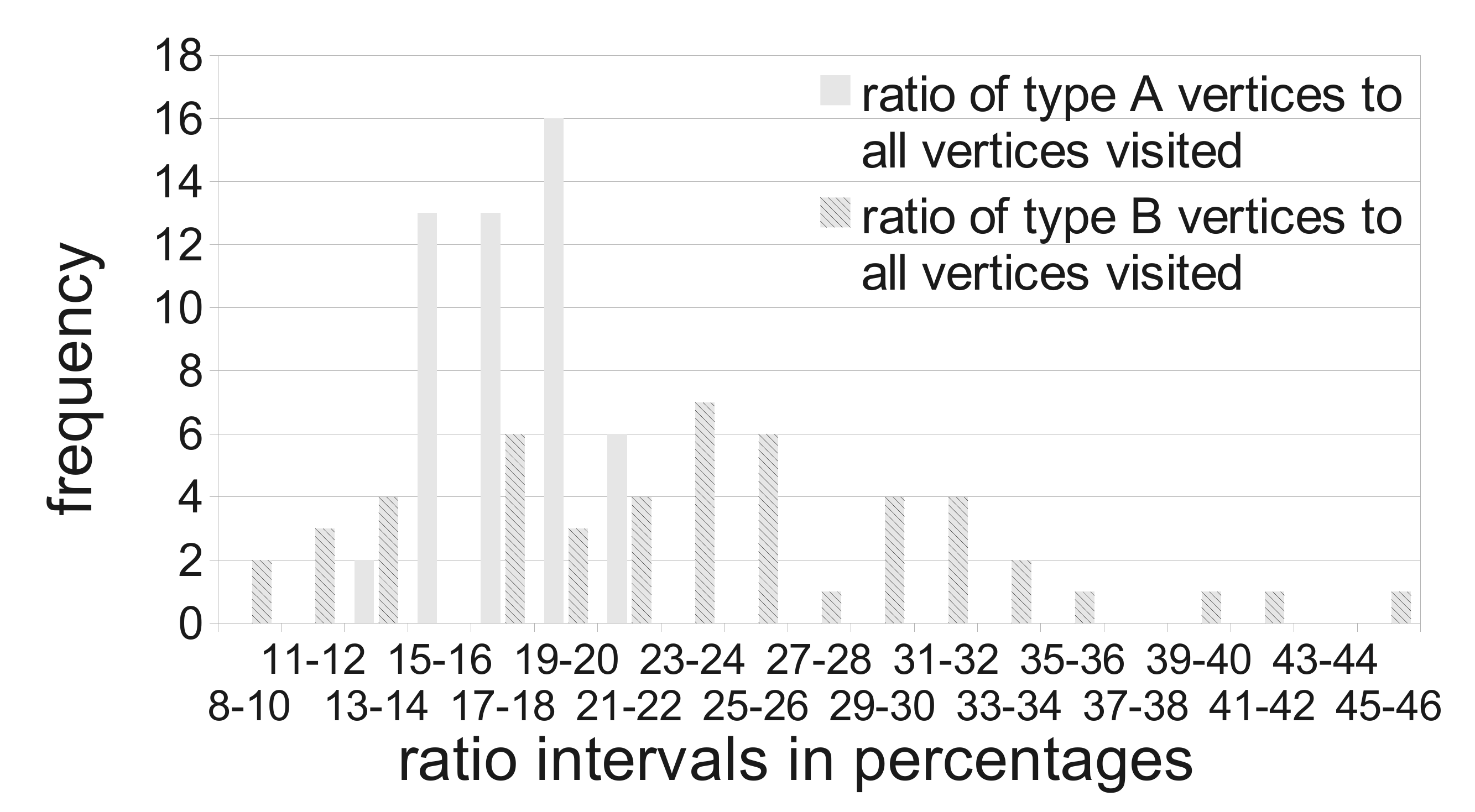}
    \captionof{figure}{Distribution of prefix tree vertices traversed for
    randomised datasets, $k_{max} = 5$, $\tau = 1$.}
    \label{fig3}
  \end{figure}

  Algorithm \ref{al:kyiv} carries out online pruning of the prefix tree so as
  to avoid walking the full prefix tree. Importantly, it also tries to avoid
  carrying out unnecessary row intersections. We can evaluate the efficiency of
  the latter by distinguishing between three types of vertices visited:
  vertices that correspond to minimal $\tau$-infrequent itemsets (A), vertices
  which are visited but for which a row intersection is not performed (B) and
  the rest of the vertices visited (C). Figure \ref{fig3} shows the
  distribution of the ratios of the number of vertices of types A and B to the
  total number of prefix tree vertices visited by the algorithm over $50$
  randomised datasets when $k_{max} = 5$. On average $17.5\%$ of the vertices
  visited are type A vertices and $23\%$ type B vertices, although sometimes up
  to $45\%$ of the vertices visited are of type $B$.

  \subsubsection{Impact of Ordering Used for \texorpdfstring{$L_{A,\tau}$}
  {Itemlist}} \label{sec:pruningperf}

  As already noted in Section \ref{sec:pruning}, the ordering used to sort set
  $L_{A,\tau}$ to obtain $L_{A,\tau}^<$ can be expected to have an impact on
  the amount of pruning of the prefix tree achieved, and so on the execution
  time of Algorithm \ref{al:kyiv}. To investigate this further, we collected
  performance measurements for three different choices of ordering: (i)
  ascending order, (ii) descending order (iii) random order (\emph{i.e.} we
  draw a permutation uniformly at random from the set of permutations mapping
  from $\{ 1, \cdots, |L_{A,\tau}| \}$ to itself and apply this permutation to
  obtain $L_{A,\tau}^<$).

  \begin{figure}
    \centering
    \includegraphics[width=0.6\textwidth]{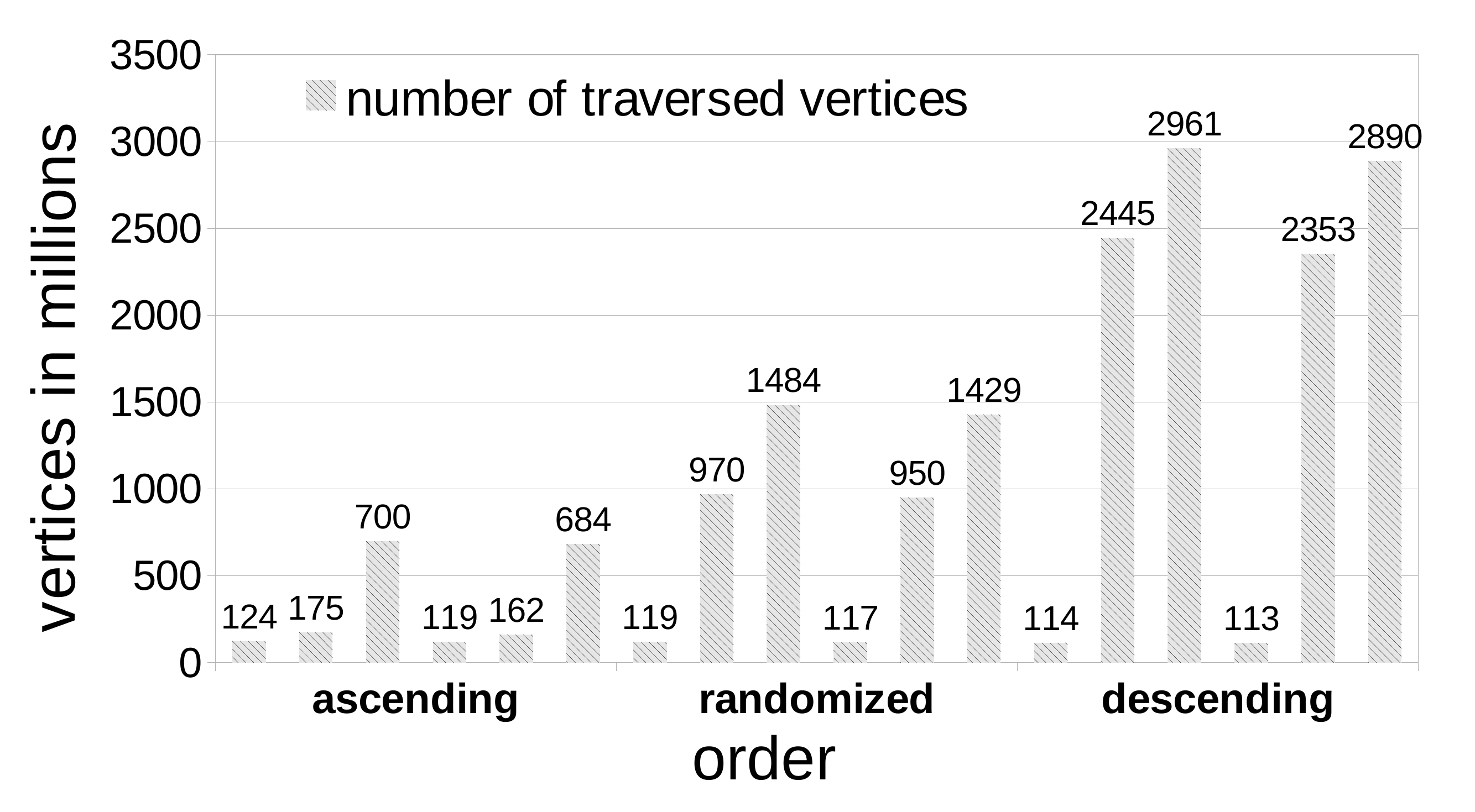}
    \captionof{figure}{Prefix tree vertices traversed vs ordering used for
    $L_{A,\tau}$, average over $10$ randomised datasets, $k_{max} = 5$, $\tau =
    1$. For each ordering $6$ values are shown: in the first three Lemma
    \ref{lem} and Corollary \ref{cor} are used, in the second three these are
    not used; in each group of three values the first value represents the
    number of vertices of type A, the second the number of vertices of type B
    and the third the total number of vertices traversed (that is of type A, B
    and C).}
    \label{fig4}
  \end{figure}

  Figure \ref{fig4} plots the numbers of prefix tree vertices of types A, B and
  C visited by Algorithm \ref{al:kyiv} vs the ordering of $L_{A,\tau}$ used. In
  this figure data is presented for each of the three orderings (ascending,
  randomised, descending) and for when Lemma \ref{lem}/Corollary \ref{cor} are
  used or not. That is, $6$ experiment variants are compared.

  \begin{figure}
    \centering
    \includegraphics[width=0.5\textwidth]{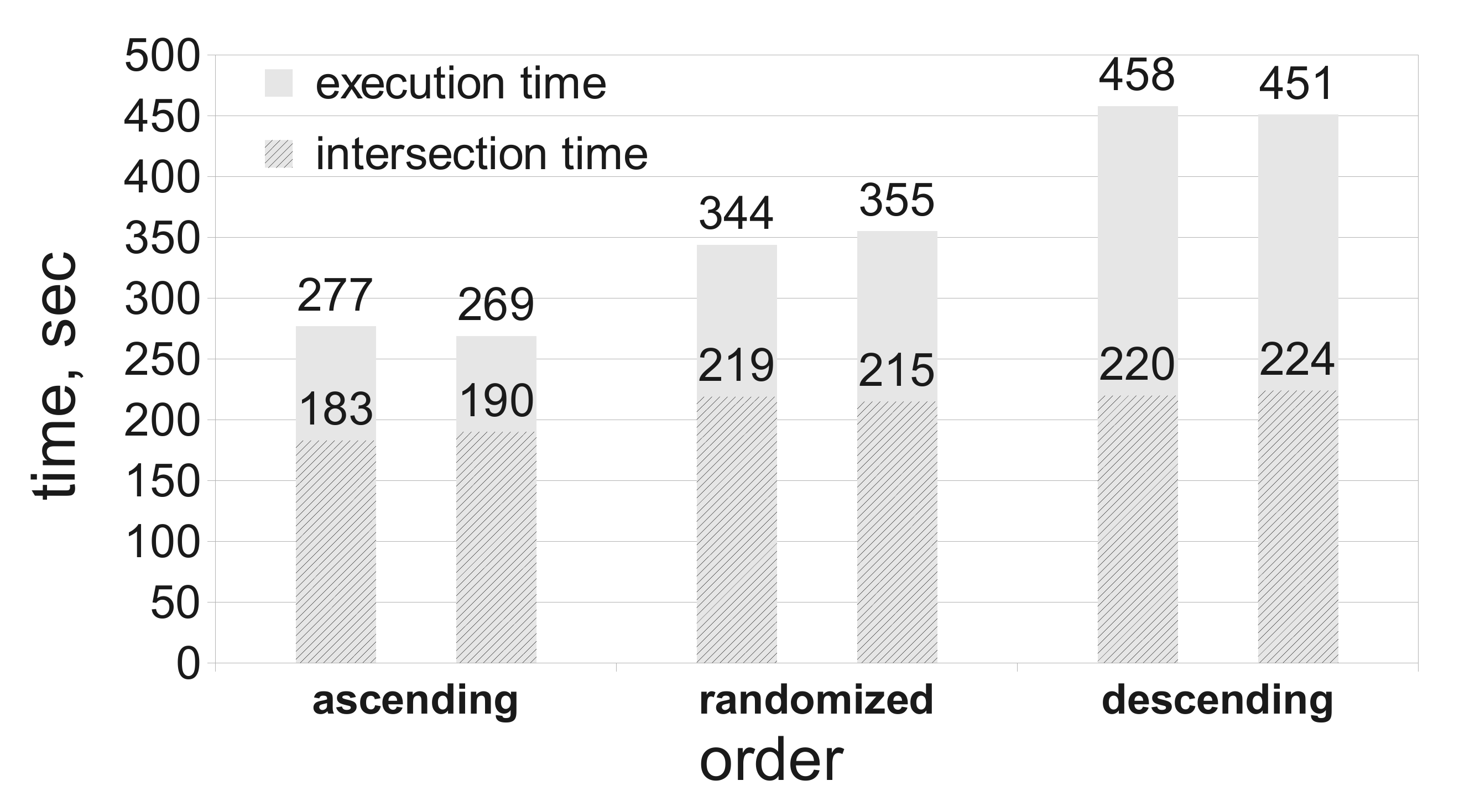}
    \captionof{figure}{Intersection and execution time vs ordering used for
    $L_{A,\tau}$, average over $10$ randomised datasets, $k_{max} = 5$, $\tau =
    2$ (in the left bar Lemma \ref{lem}/Corollary \ref{cor} are used, in the
    right bar they are not used).}
    \label{fig5}
  \end{figure}

  It can be seen that use of ascending order significantly reduces the total
  number of vertices visited, yielding a reduction of roughly a factor of $2$
  compared to use of a randomised ordering and a factor of $4$ compared to
  descending order. The number of type A vertices visited is, as expected,
  essentially constant across the tests. However, the number of type B vertices
  changes significantly and varies such that the number of vertices of type C
  remains roughly constant. Observe that use of Lemma \ref{lem} and Corollary
  \ref{cor} has little impact on performance in these tests. We will revisit
  this in Section \ref{sec:specdats} where we find that they can speed the
  runtime up by more than $50\%$.

  Figure \ref{fig5} plots the corresponding intersection and execution time vs
  the ordering of $L_{A,\tau}$ used. It can be seen that the execution time is
  more sensitive to the ordering than the intersection time. When combined with
  Figure \ref{fig4} this allows us to conclude that it is the number of type B
  vertices that varies strongly with ordering (the number of type A and type C
  vertices stays nearly constant) and that ascending order reduces execution
  time primarily by reducing the number of type B vertices i.e. by more
  effective pruning of the search tree which reduces the overall number of
  vertices visited.

  \subsubsection{Impact of Dataset Parameters}

  To investigate the scaling behaviour of Algorithm \ref{al:kyiv} to larger
  datasets we generated a randomised dataset with $1,000,000$ rows and $40$
  columns yielding an itemlist of size $2,179$.

  \begin{figure}
    \centering
    \subfloat[$m = 40$]{
      \includegraphics[width=0.485\textwidth]
      {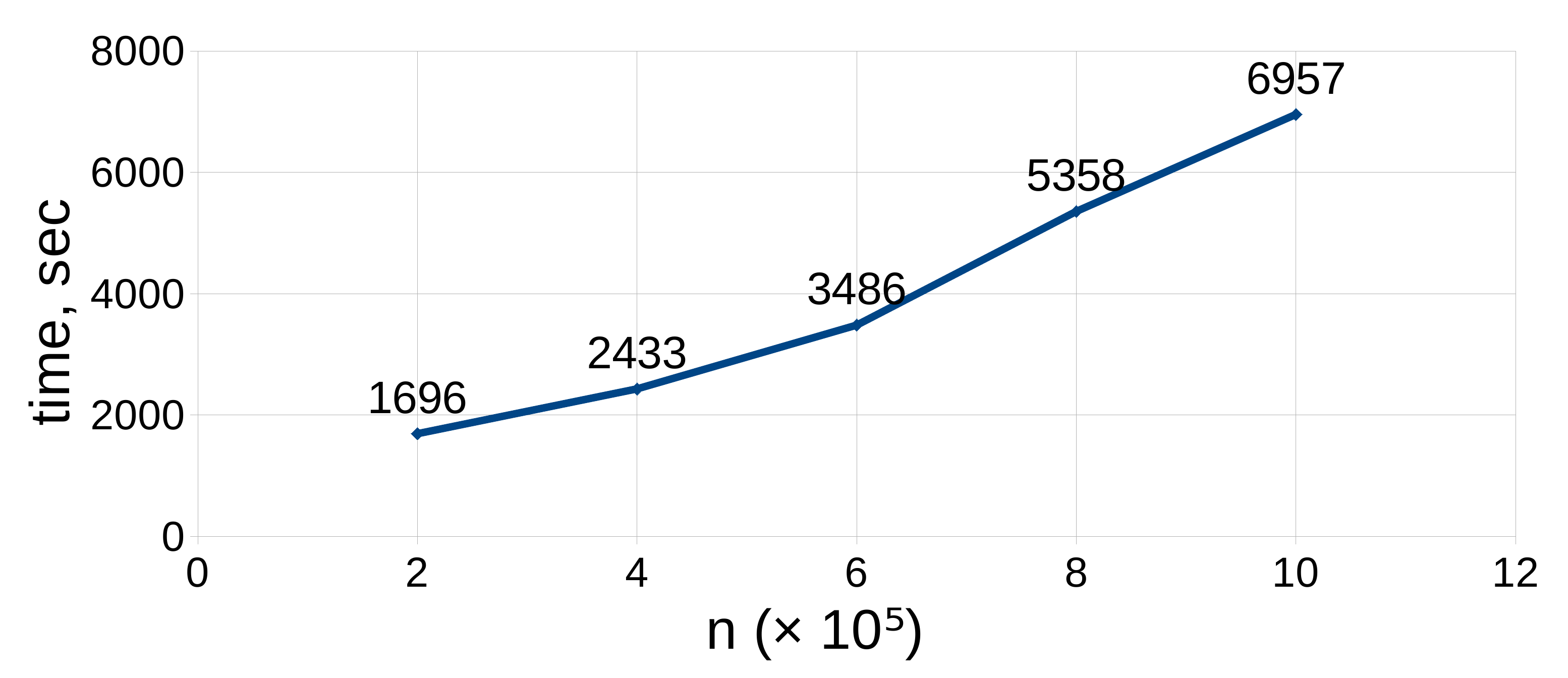}
      \label{fig6}
    }
    \subfloat[$n = 1,000,000$]{
      \includegraphics[width=0.485\textwidth]
      {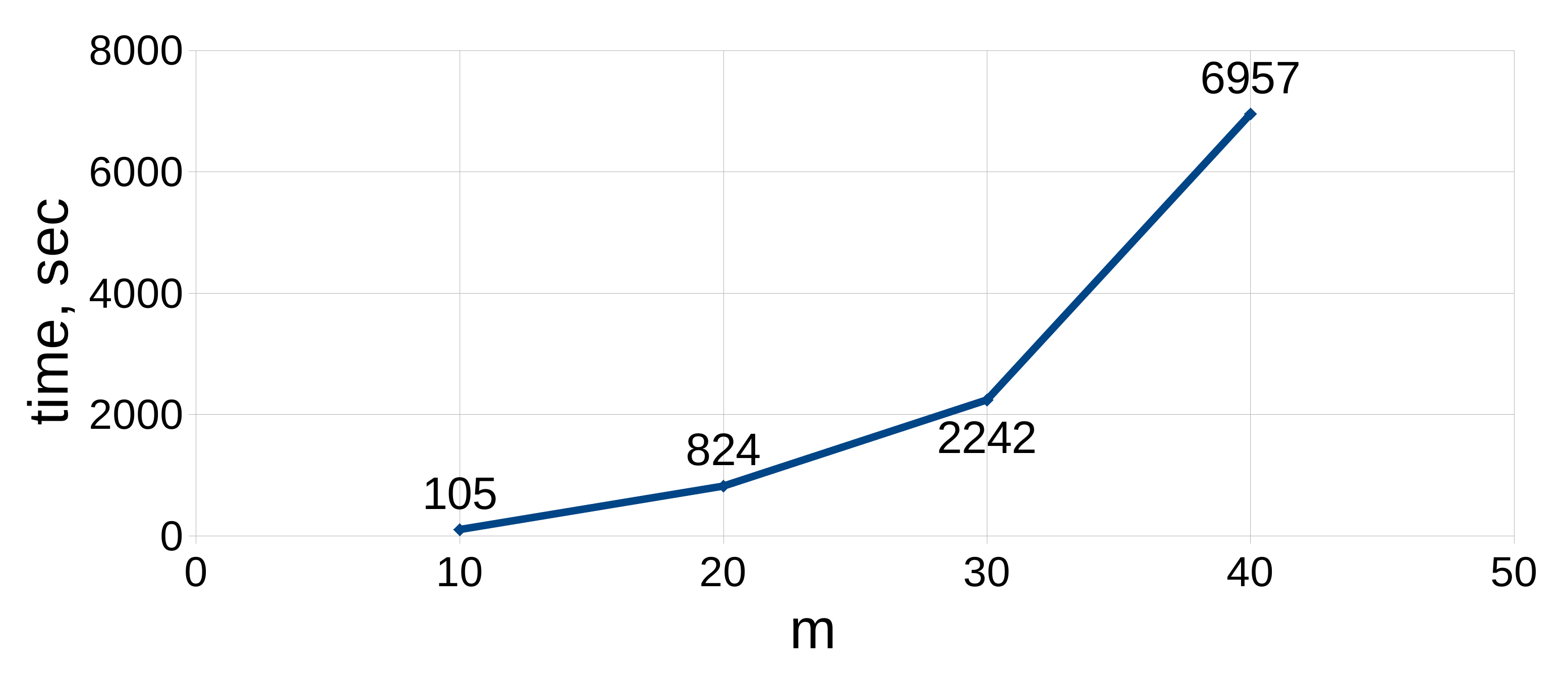}
      \label{fig7}
    }
    \captionof{figure}{Execution time vs number of rows $n$ and columns $m$ for
    a randomised dataset, $k_{max} = 3$, $\tau = 1$.}
  \end{figure}

  Taking the first $n$ rows, Figure \ref{fig6} plots the execution time of
  Algorithm \ref{al:kyiv} versus $n$ for $k_{max} = 3$, $\tau = 1$. It can be
  seen that the execution time is approximately linear in $n$, and so scales
  well to larger datasets. Although not plotted, memory usage also increased
  only gradually from $5.6$Gb when $n = 200,000$ to $6$Gb when $n = 1,000,000$.

  Taking the first $m$ columns of the dataset, Figure \ref{fig7} plots the
  execution time versus $m$ for $k_{max} = 3$, $\tau = 1$. It can be seen that
  the execution time is approximately exponential in $m$, and so the algorithm
  scales less well to datasets with a large number of columns (the size of
  corresponding itemlist increased from $520$ to $2,179$). Note that the memory
  usage also increases quite rapidly with $m$, from $0.9$Gb when $m = 10$ to
  $6$Gb when $m = 40$.

  \subsection{Domain-Specific Performance}

  \subsubsection{Datasets} \label{sec:datasets}

  In this section we present performance measurements for four domain-specific
  datasets:
  \begin{enumerate}
    \item The Connect dataset is available from \url{http://fimi.ua.ac.be/data}
    and contains all legal 8-ply positions in the game of connect-4 in which
    neither player has won yet, and in which the next move is not forced. There
    are $67,557$ rows, $43$ columns (one for each of the 42 connect-4 squares
    together with an outcome column - win, draw or lose) and $129$ items. It
    was one of the most computationally challenging datasets for which MINIT
    was evaluated in \cite{minit}.
    \item The Pumsb dataset is census data for population and housing from the
    PUMS (Public Use Microdata Sample). This dataset is available from
    \url{http://fimi.ua.ac.be/data}. There are $49,046$ rows, $74$ columns and
    $1,958$ items.
    \item The Poker dataset is available from
    \url{http://archive.ics.uci.edu/ml/datasets.html}. Each record is an
    example of a hand consisting of five playing cards drawn from a standard
    deck of 52 cards. Each card is described using two attributes (suit and
    rank), for a total of 10 predictive attributes. There is one Class
    attribute that describes the "Poker Hand". We removed the last attribute to
    form a new dataset with $1,000,000$ rows, $10$ columns and $117$ items.
    \item The USCensus1990 dataset, available from
    \url{http://archive.ics.uci.edu/ml/datasets.html}, was collected as part of
    the 1990 census. We considered a subset of this dataset consisting of the
    first $200,000$ rows and $68$ columns, which contained $8,009$ items.
  \end{enumerate}

  \subsubsection{Execution Time vs \texorpdfstring{$k_{max}$}{Maximum Itemset
  Size}} \label{sec:specdats}

  All measurements in the current section are averaged over three consecutive
  runs of each algorithm.

  \begin{figure}
    \centering
    \subfloat[$\tau = 1$]{
      \includegraphics[width=0.485\textwidth]
      {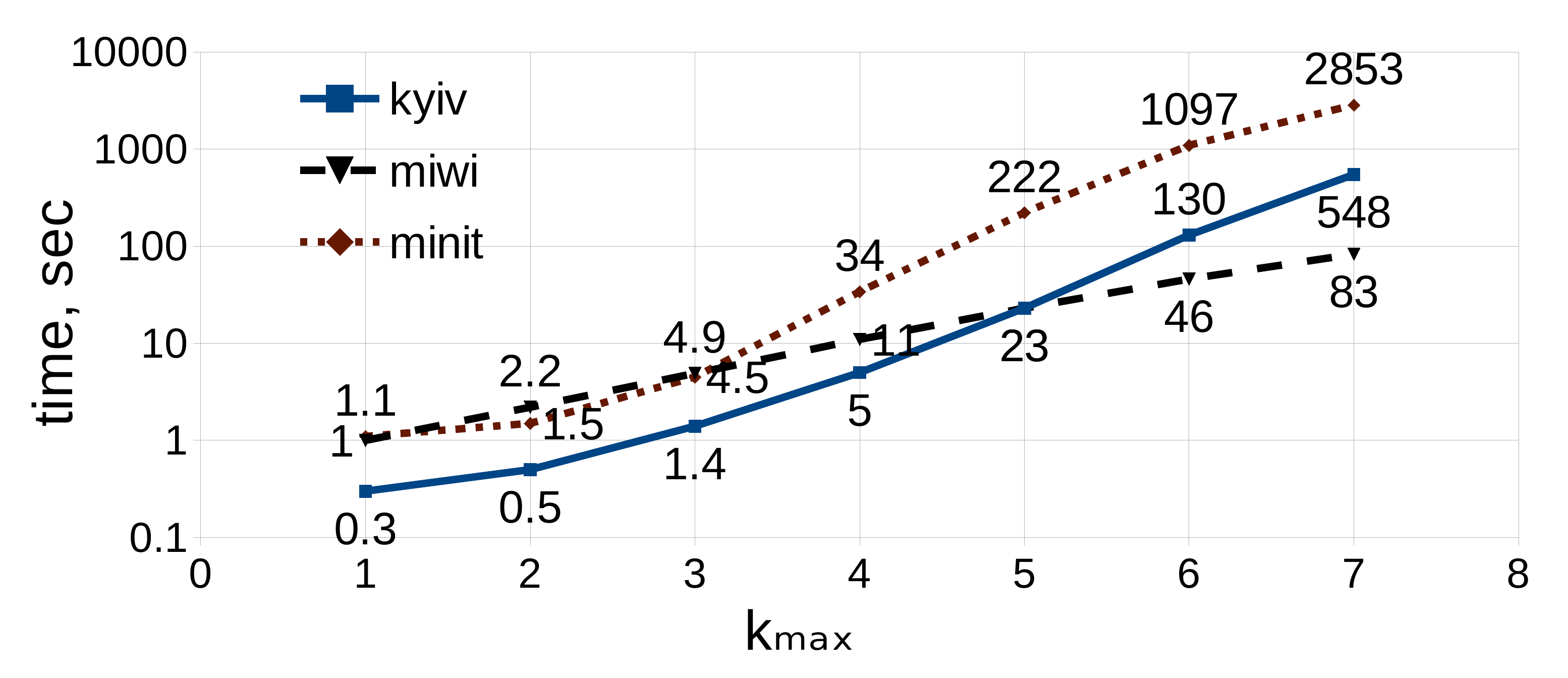}
      \label{fig8}
    }
    \subfloat[$\tau = 5$]{
      \includegraphics[width=0.485\textwidth]
      {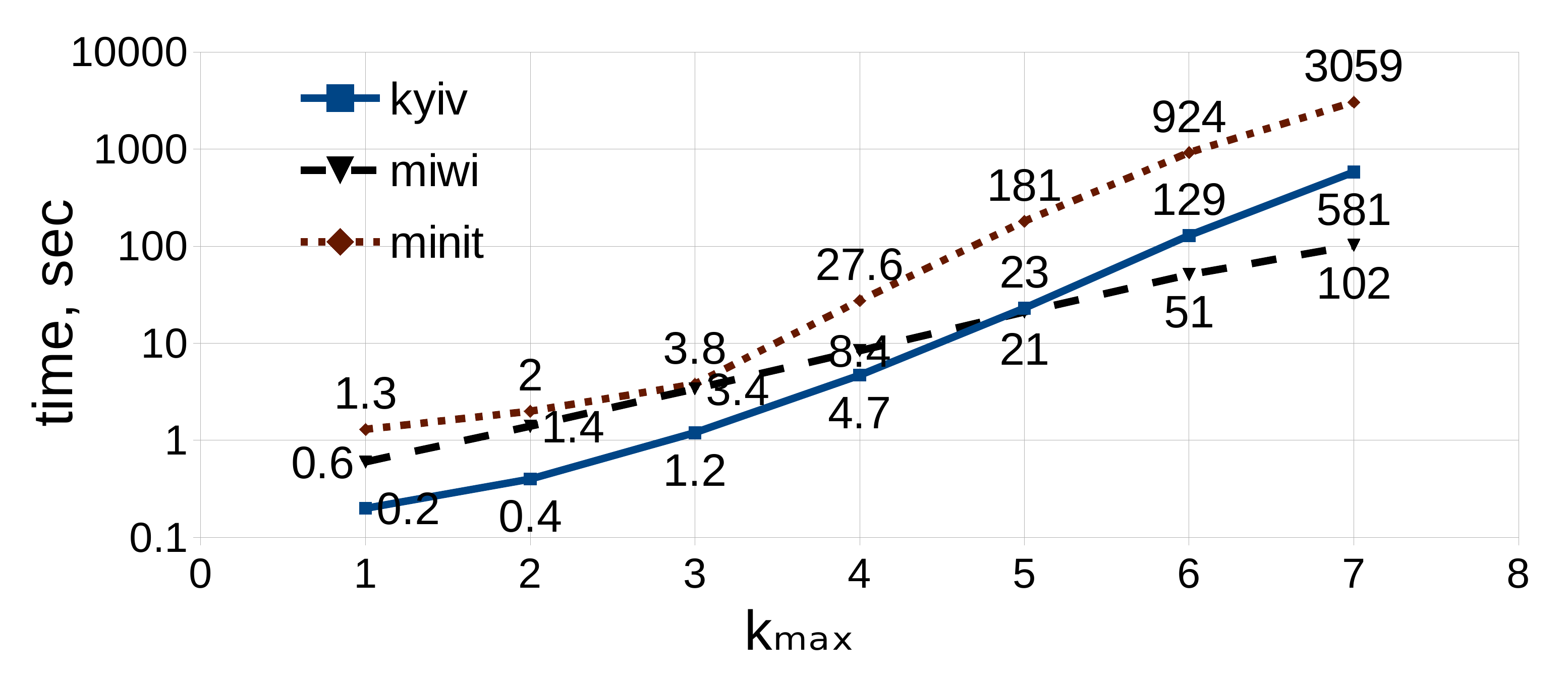}
      \label{fig9}
    }\\
    \subfloat[$\tau = 10$]{
      \includegraphics[width=0.485\textwidth]
      {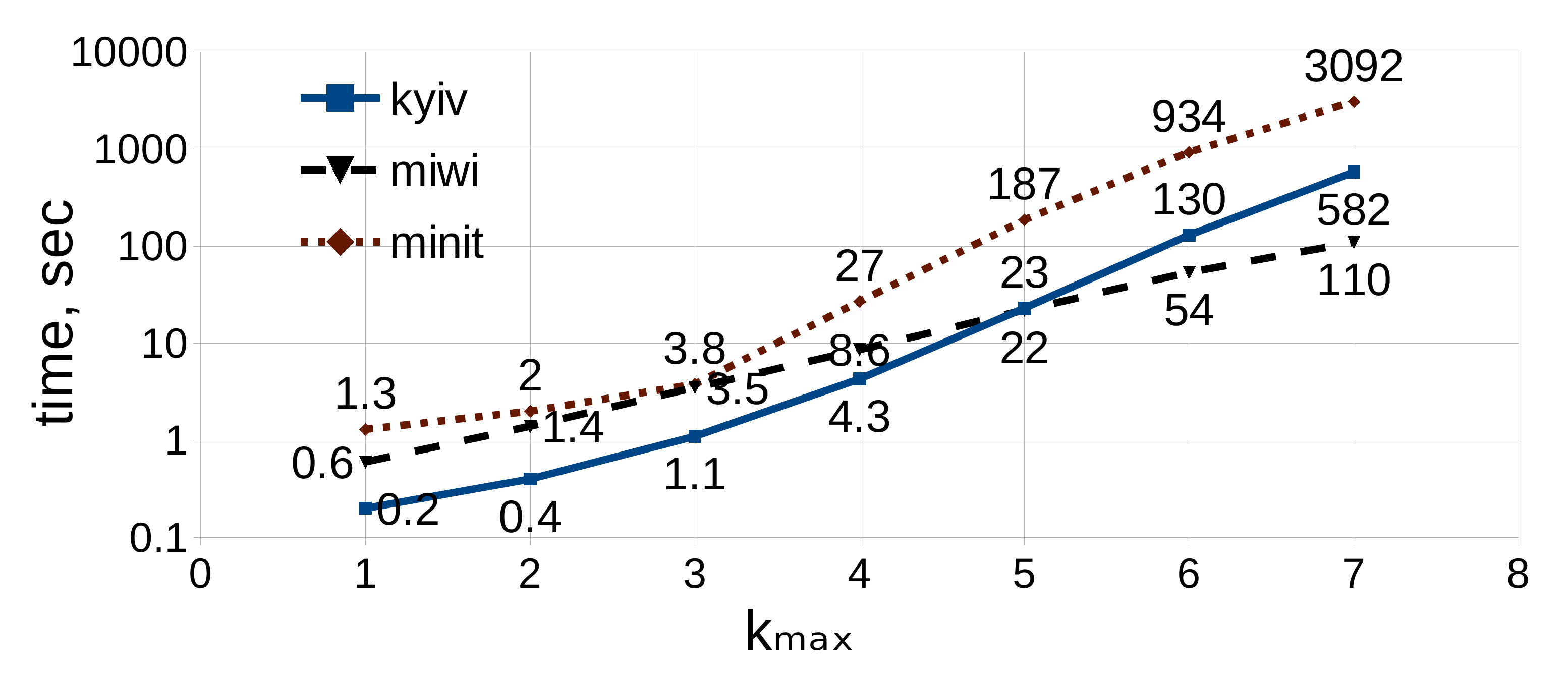}
      \label{fig10}
    }
    \subfloat[$\tau = 100$]{
      \includegraphics[width=0.485\textwidth]
      {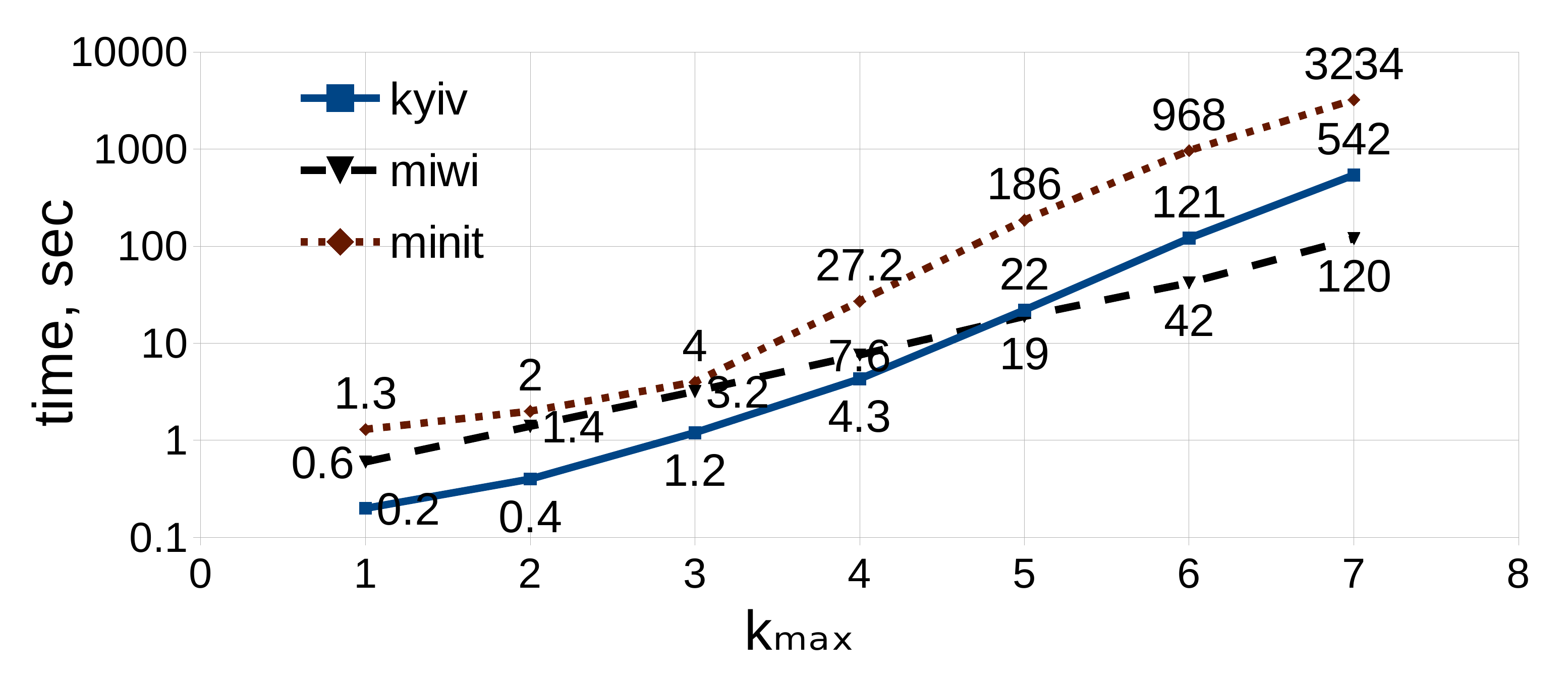}
      \label{fig11}
    }
    \captionof{figure}{Execution time vs $k_{max}$ for Connect dataset.}
    \label{fig8all}
  \end{figure}

  \begin{figure}
    \centering
    \subfloat[$\tau = 1$]{
      \includegraphics[width=0.485\textwidth]
      {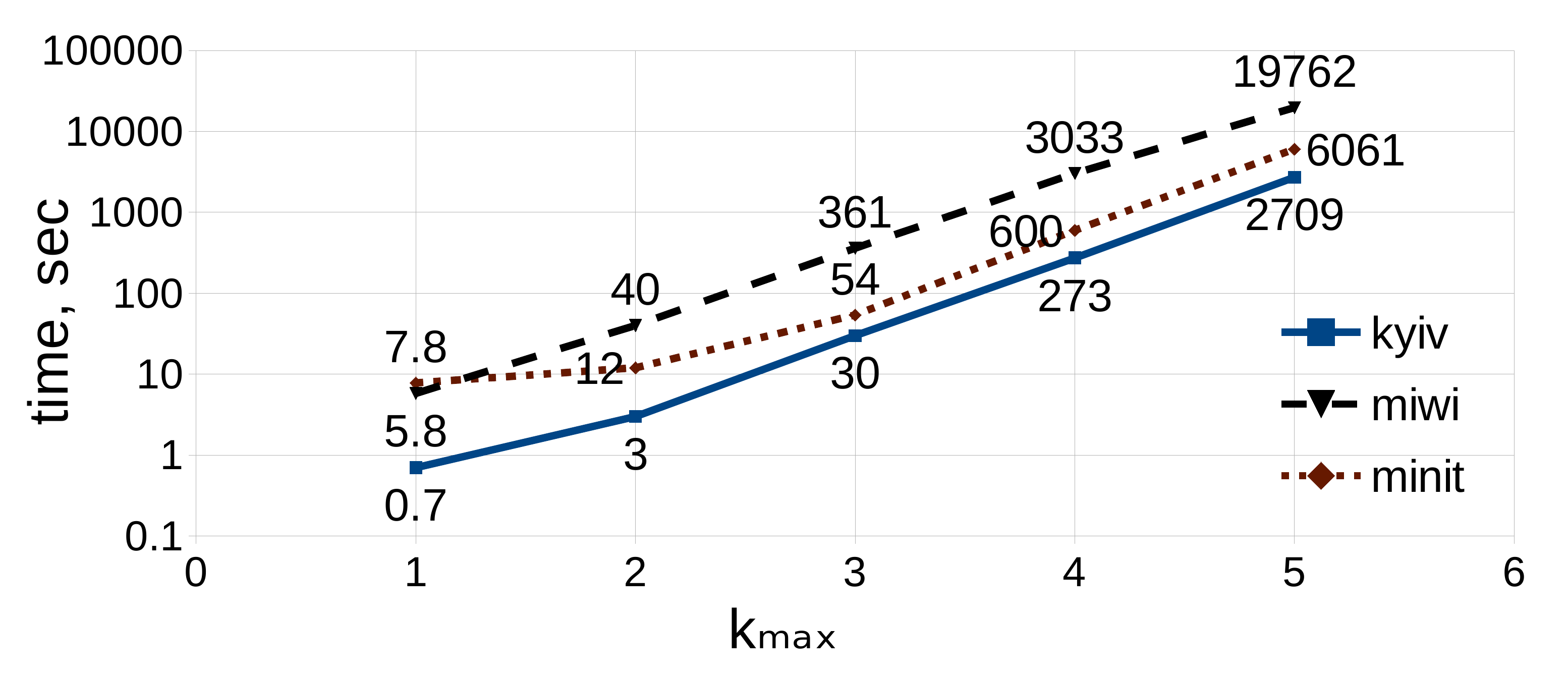}
      \label{fig12}
    }
    \subfloat[$\tau = 5$]{
      \includegraphics[width=0.485\textwidth]
      {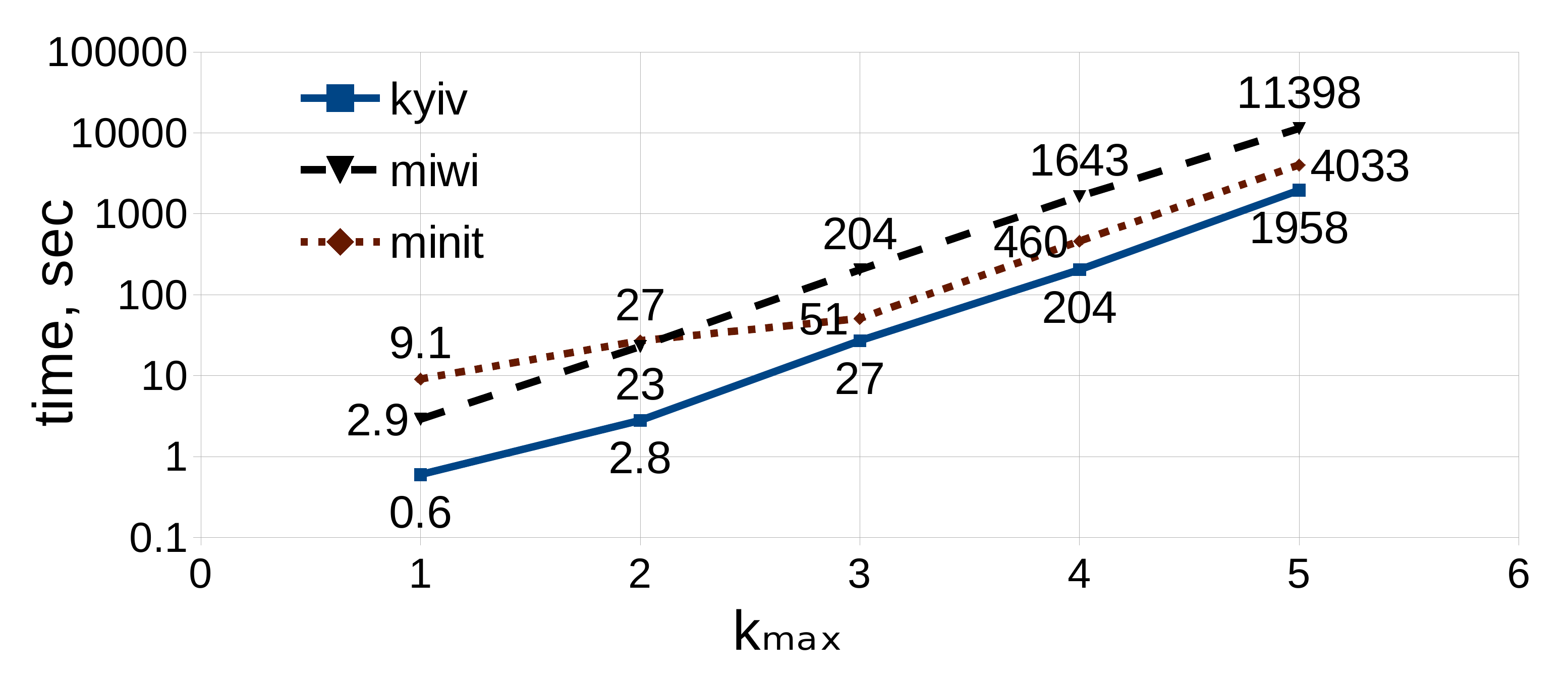}
      \label{fig13}
    }\\
    \subfloat[$\tau = 10$]{
      \includegraphics[width=0.485\textwidth]
      {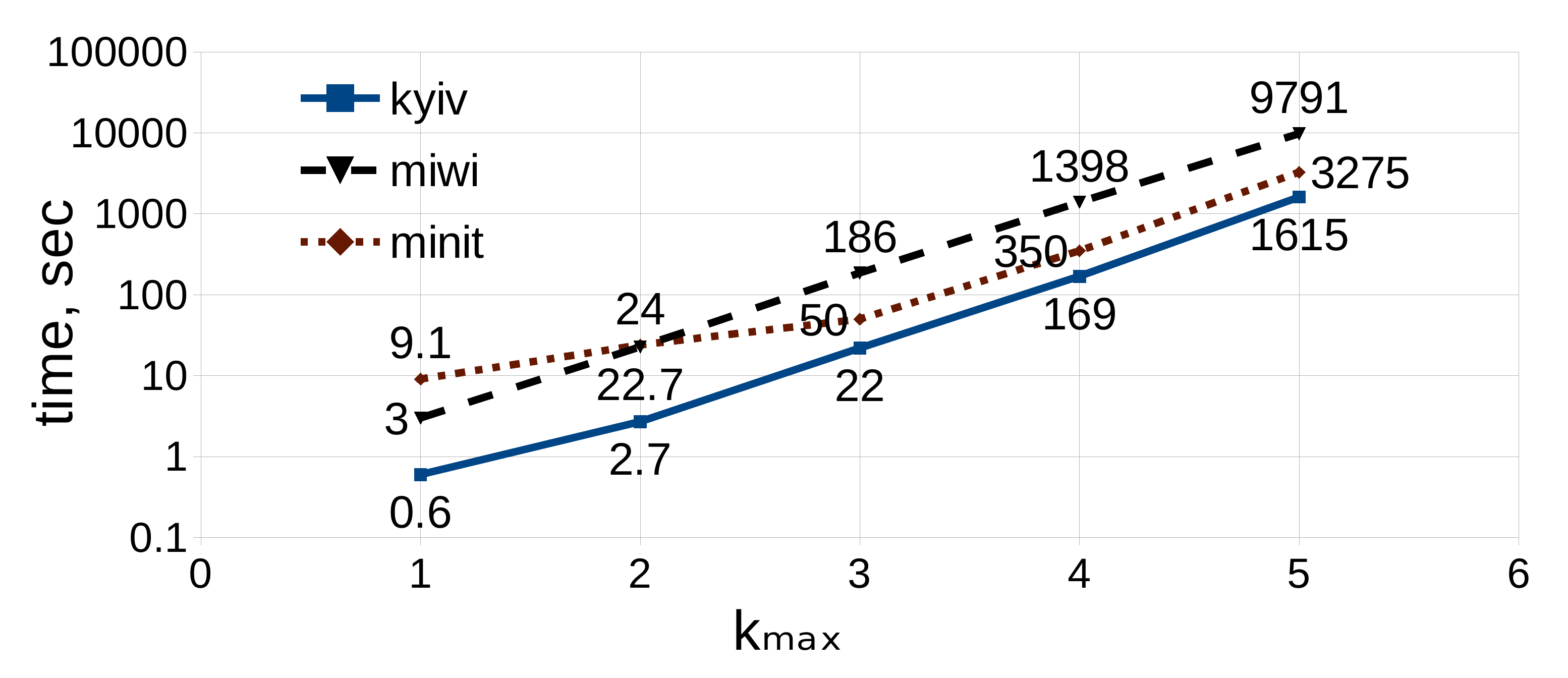}
      \label{fig14}
    }
    \subfloat[$\tau = 100$]{
      \includegraphics[width=0.485\textwidth]
      {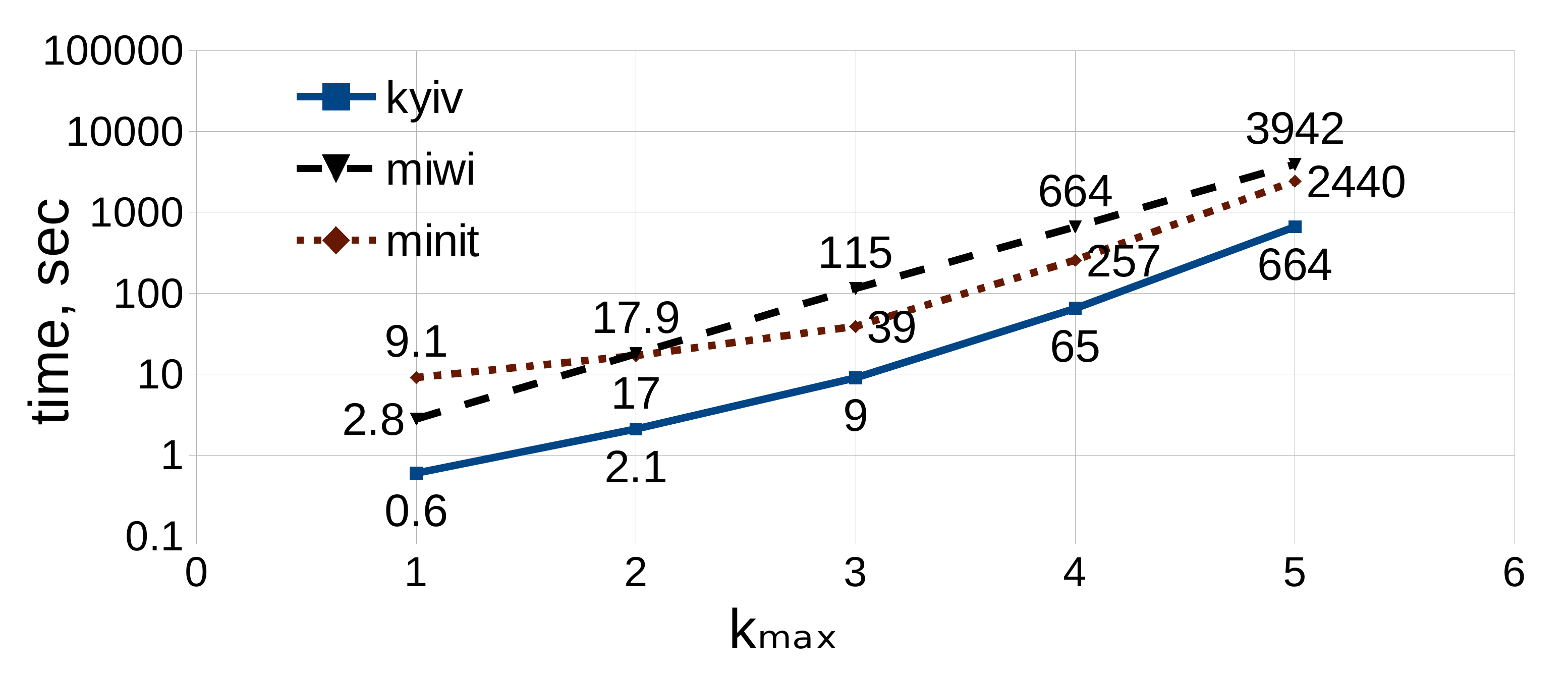}
      \label{fig15}
    }
    \captionof{figure}{Execution time vs $k_{max}$ for Pumsb dataset.}
    \label{fig12all}
  \end{figure}

  \begin{figure}
    \centering
    \subfloat[$\tau = 1$]{
      \includegraphics[width=0.485\textwidth]
      {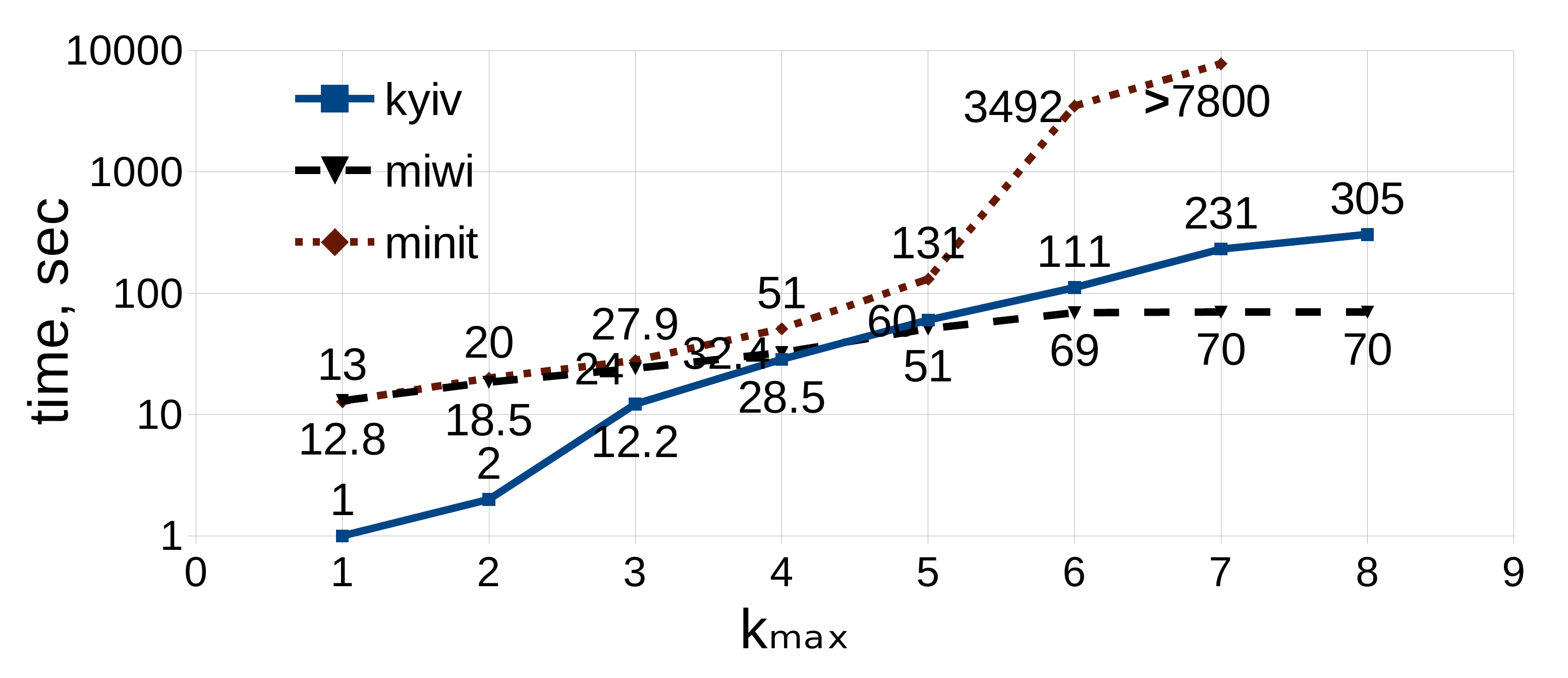}
      \label{fig16}
    }
    \subfloat[$\tau = 5$]{
      \includegraphics[width=0.485\textwidth]
      {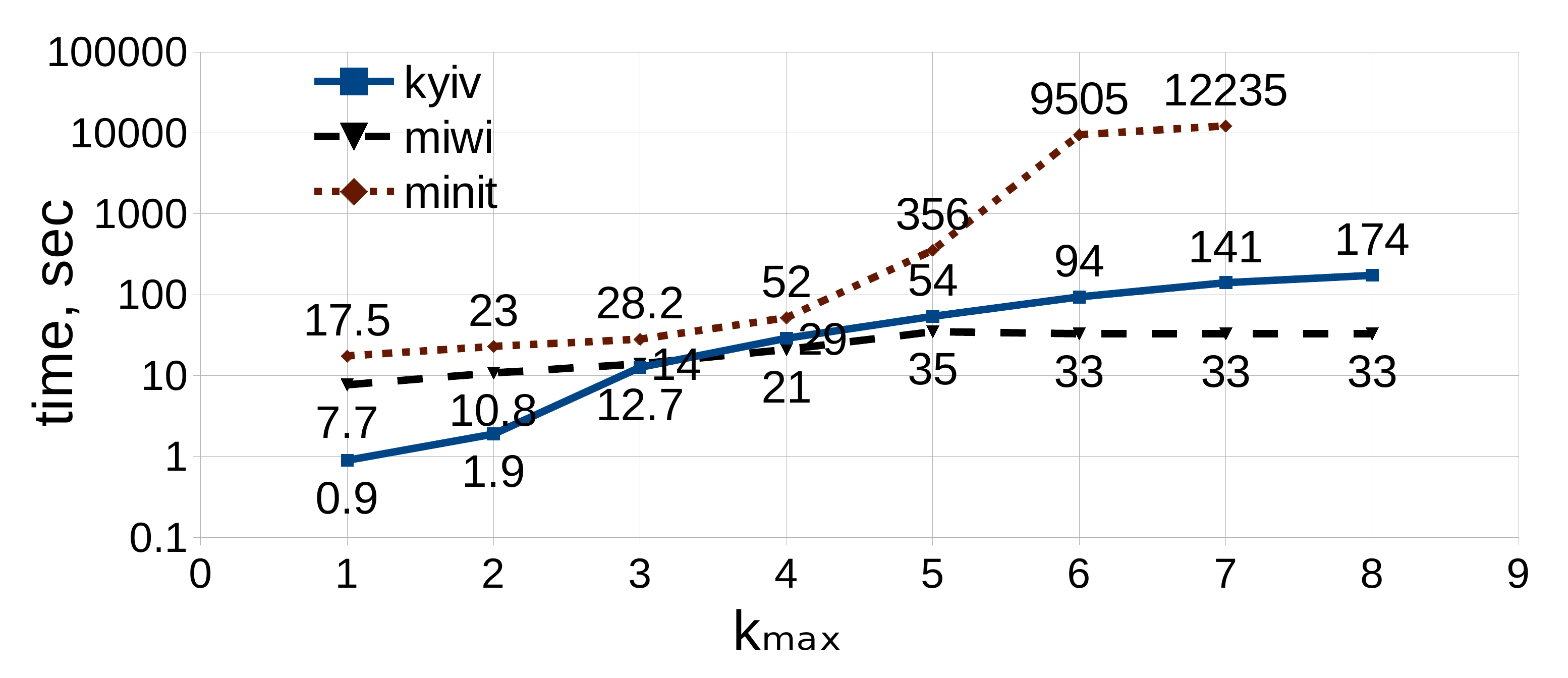}
      \label{fig17}
    }\\
    \subfloat[$\tau = 10$]{
      \includegraphics[width=0.485\textwidth]
      {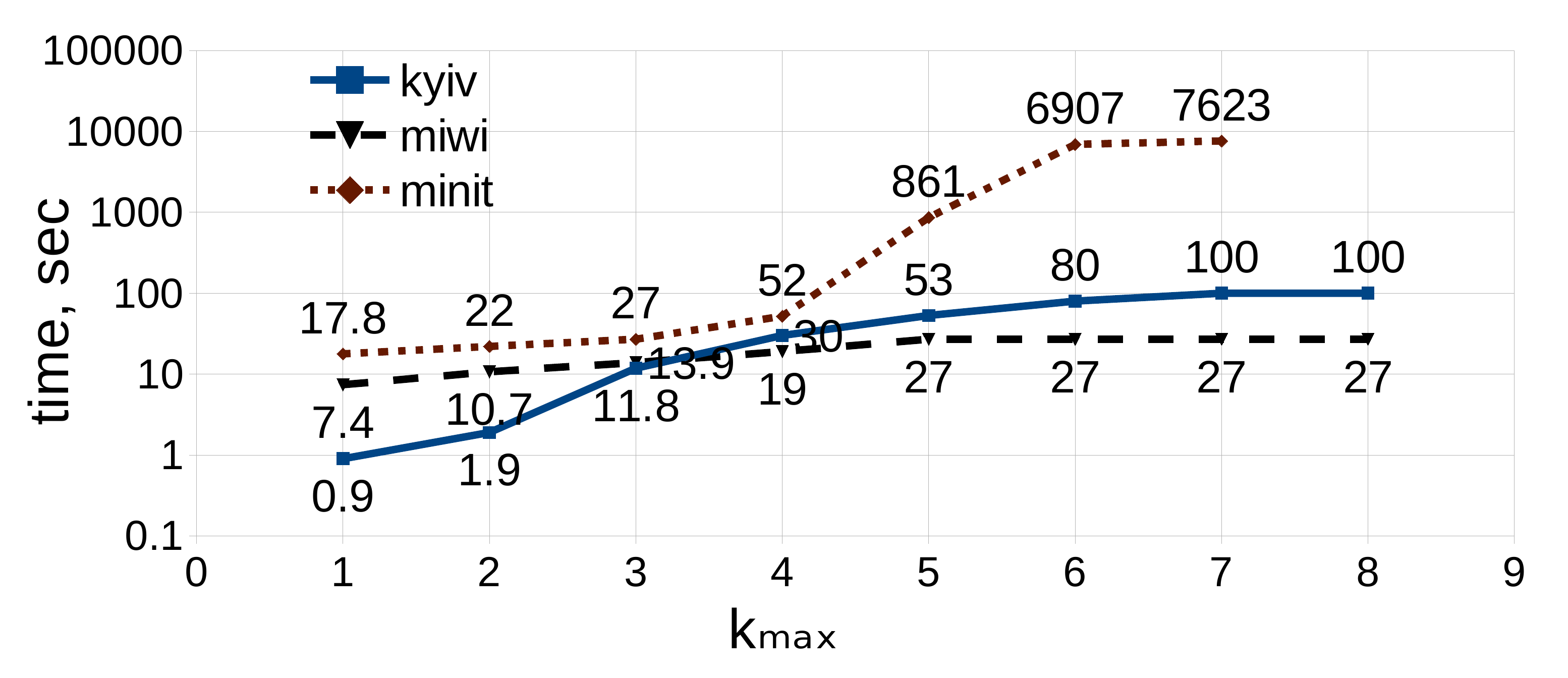}
      \label{fig18}
    }
    \subfloat[$\tau = 100$]{
      \includegraphics[width=0.485\textwidth]
      {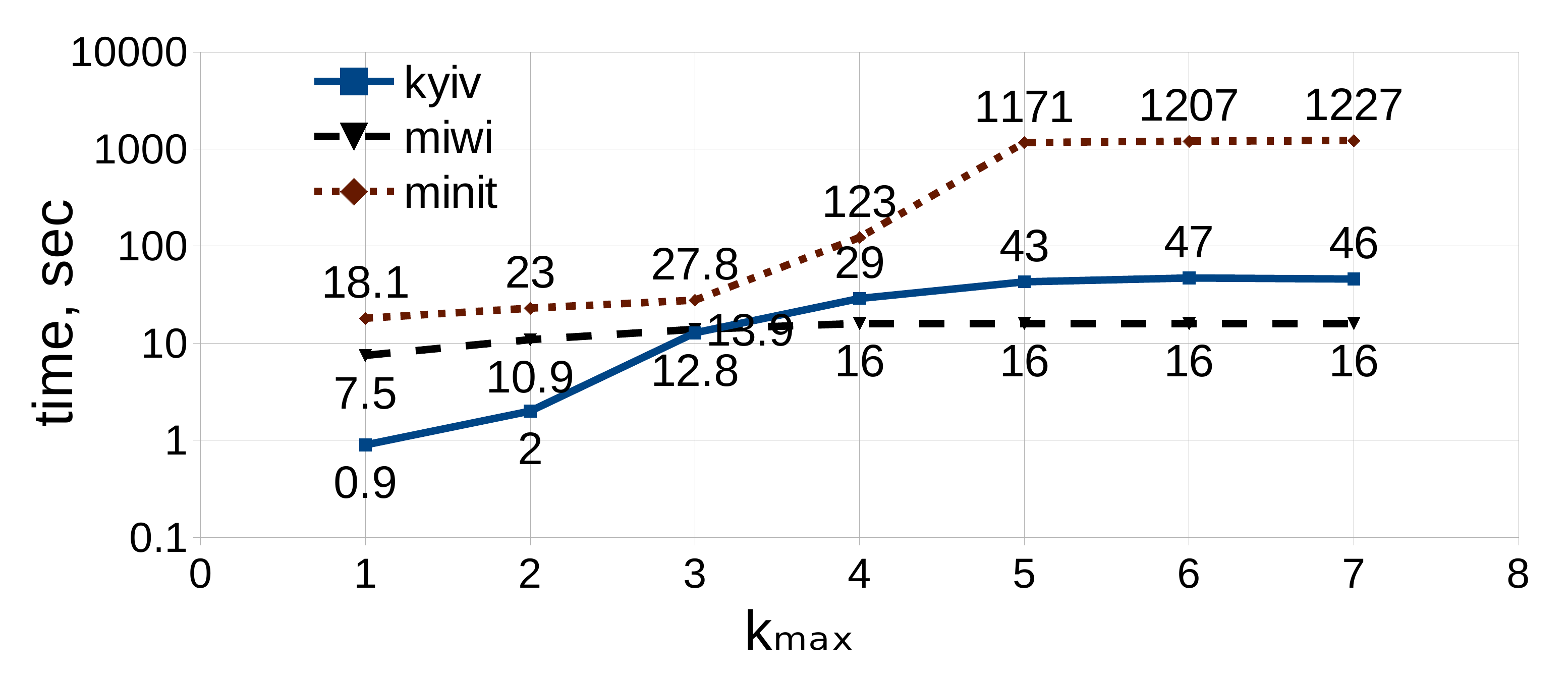}
      \label{fig19}
    }
    \captionof{figure}{Execution time vs $k_{max}$ for Poker dataset.}
    \label{fig16all}
  \end{figure}

  \begin{figure}
    \centering
    \subfloat[$\tau = 1$]{
      \includegraphics[width=0.485\textwidth]
      {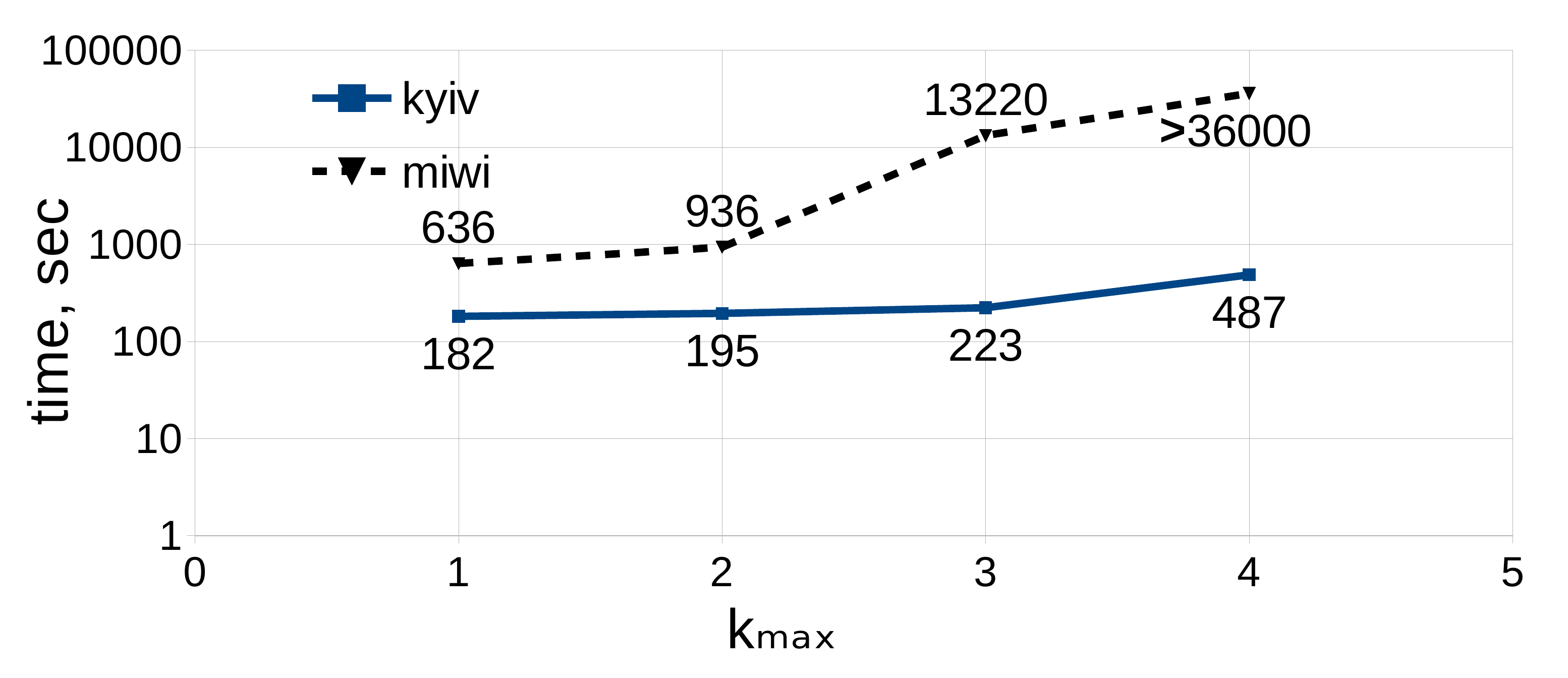}
      \label{fig20}
    }
    \subfloat[$\tau = 5$]{
      \includegraphics[width=0.485\textwidth]
      {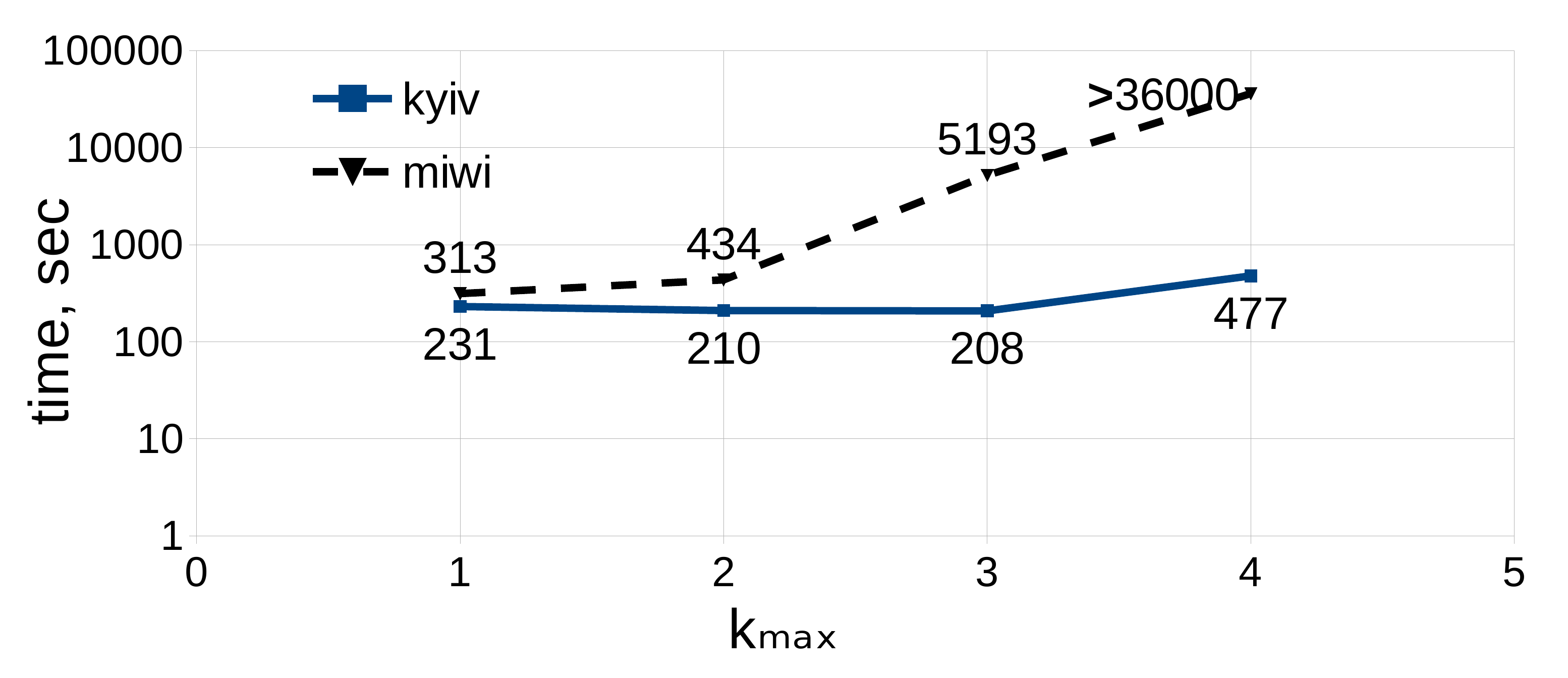}
      \label{fig21}
    }\\
    \subfloat[$\tau = 10$]{
      \includegraphics[width=0.485\textwidth]
      {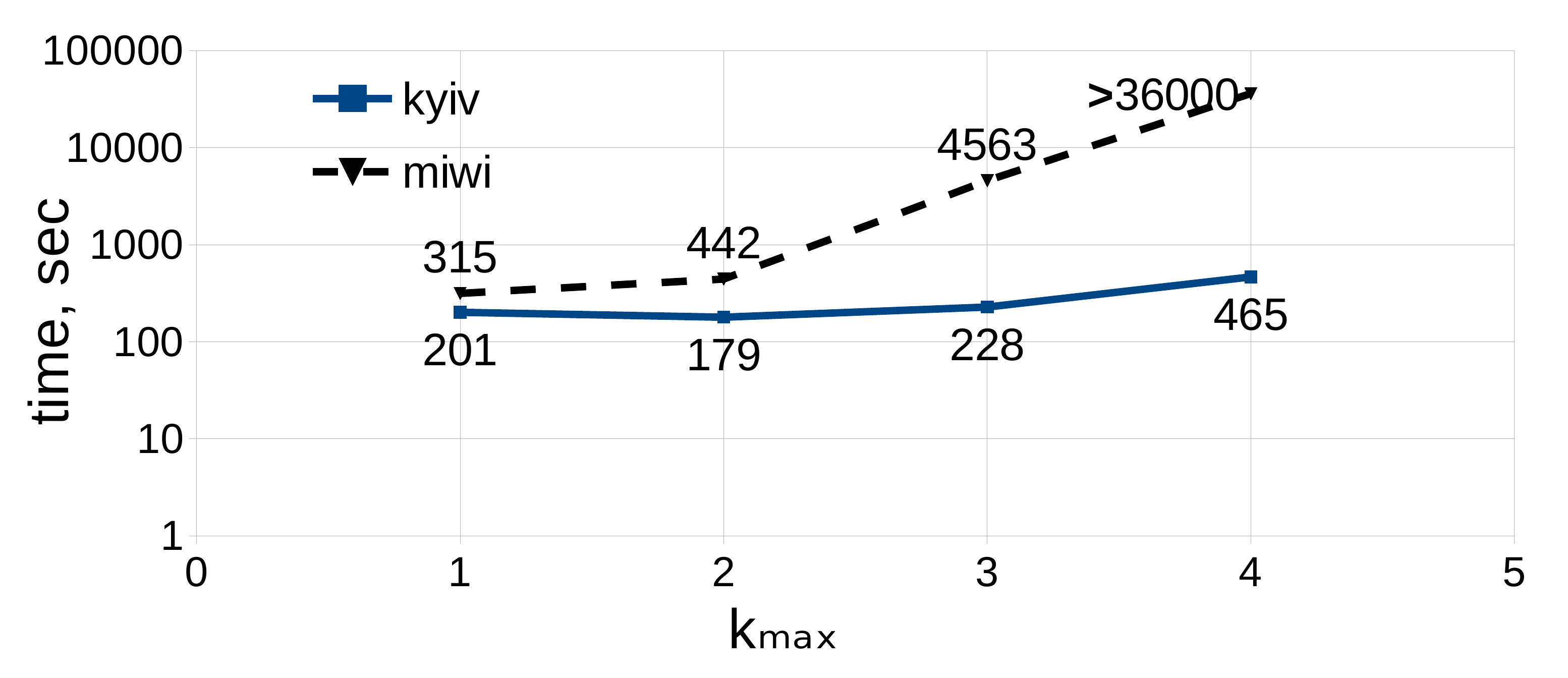}
      \label{fig22}
    }
    \subfloat[$\tau = 100$]{
      \includegraphics[width=0.485\textwidth]
      {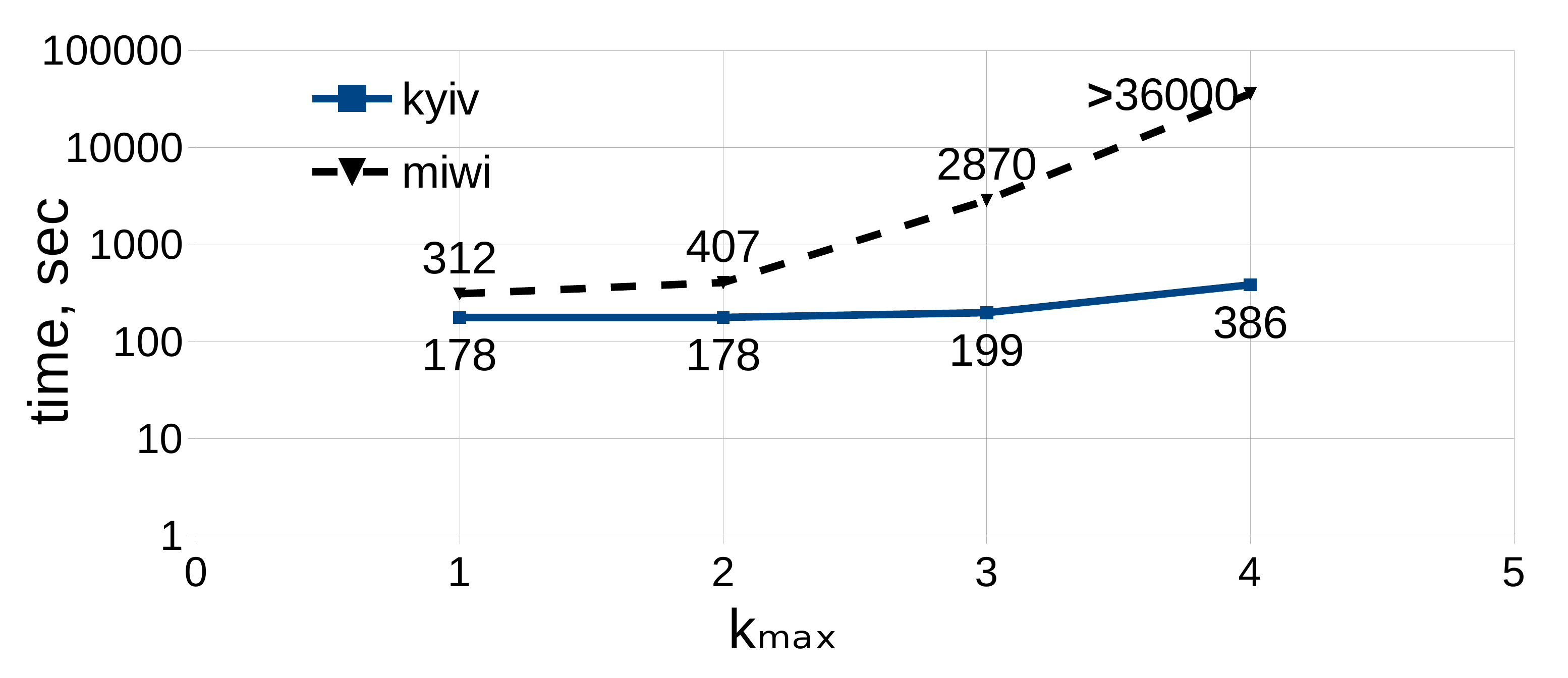}
      \label{fig23}
    }
    \captionof{figure}{Execution time vs $k_{max}$ for USCensus1990 dataset.}
    \label{fig20all}
  \end{figure}

  Figures \ref{fig8all}, \ref{fig12all}, \ref{fig16all} and \ref{fig20all} show
  the measured execution times of Algorithm \ref{al:kyiv}, MINIT and MIWI Miner
  measured for the Connect, Pumsb, Poker and USCensus1990 datasets vs $k_{max}$
  when $\tau = 1$, $5$, $10$ and $100$.

  It can be seen that Algorithm \ref{al:kyiv} consistently outperforms MINIT
  for all values of $k_{max}$ and $\tau$ and for all datasets. For the Connect
  dataset it can be seen that Algorithm \ref{al:kyiv} achieves runtimes between
  $3$ and $9$ times faster than MINIT. For the Pumsb dataset Algorithm
  \ref{al:kyiv} is between $2$ and $11$ times faster. For the Poker dataset
  Algorithm \ref{al:kyiv} is between $2$ and $33$ times faster (for $k_{max} =
  7$, $\tau = 1$ MINIT was terminated after $7,800$ seconds without
  completing). Data is not shown for the USCensus1990 dataset since both the C++
  and Java implementations of MINIT ran out of memory on this demanding dataset
  (which has $8,009$ items).

  For the Connect and Poker datasets MIWI is $2-7$ times faster than Algorithm
  \ref{al:kyiv} when $k_{max} > 4$, but MIWI is $2-9$ times slower than
  Algorithm \ref{al:kyiv} when $k_{max} \le 4$. MIWI is also $5-13$ times
  slower than Algorithm \ref{al:kyiv} for the Pumsb dataset for all values of
  $k_{max}$ (and also slower than MINIT for this dataset). For the demanding
  USCensus1990 dataset MIWI's execution time is $220$ minutes when
  $k_{max} = 3$, $\tau = 1$ and it did not complete within a reasonable time
  for $k_{max} = 4$. In comparison, Algorithm \ref{al:kyiv} finds minimal
  sample uniques for $k_{max} = 4$ in $8$ minutes while for $k_{max} = 3$ the
  execution time reduces to $3$ minutes.

  Revisiting the order analysis in Section \ref{sec:pruningperf}, we point out
  that when Algorithm \ref{al:kyiv} is run without using Lemma \ref{lem} and
  Corollary \ref{cor} then the execution time rises to $269$ seconds (from
  $130$ seconds) for the Connect dataset, $k_{max} = 6$ and to $410$ (from
  $273$ seconds) seconds for the Pumsb dataset, $k_{max} = 4$ for example.

  \subsubsection{Execution Time vs \texorpdfstring{$\tau$}{Frequency
  Threshold}}

  \begin{figure}
    \centering
    \subfloat[Connect, $k_{max} = 7$]{
      \includegraphics[width=0.485\textwidth]
      {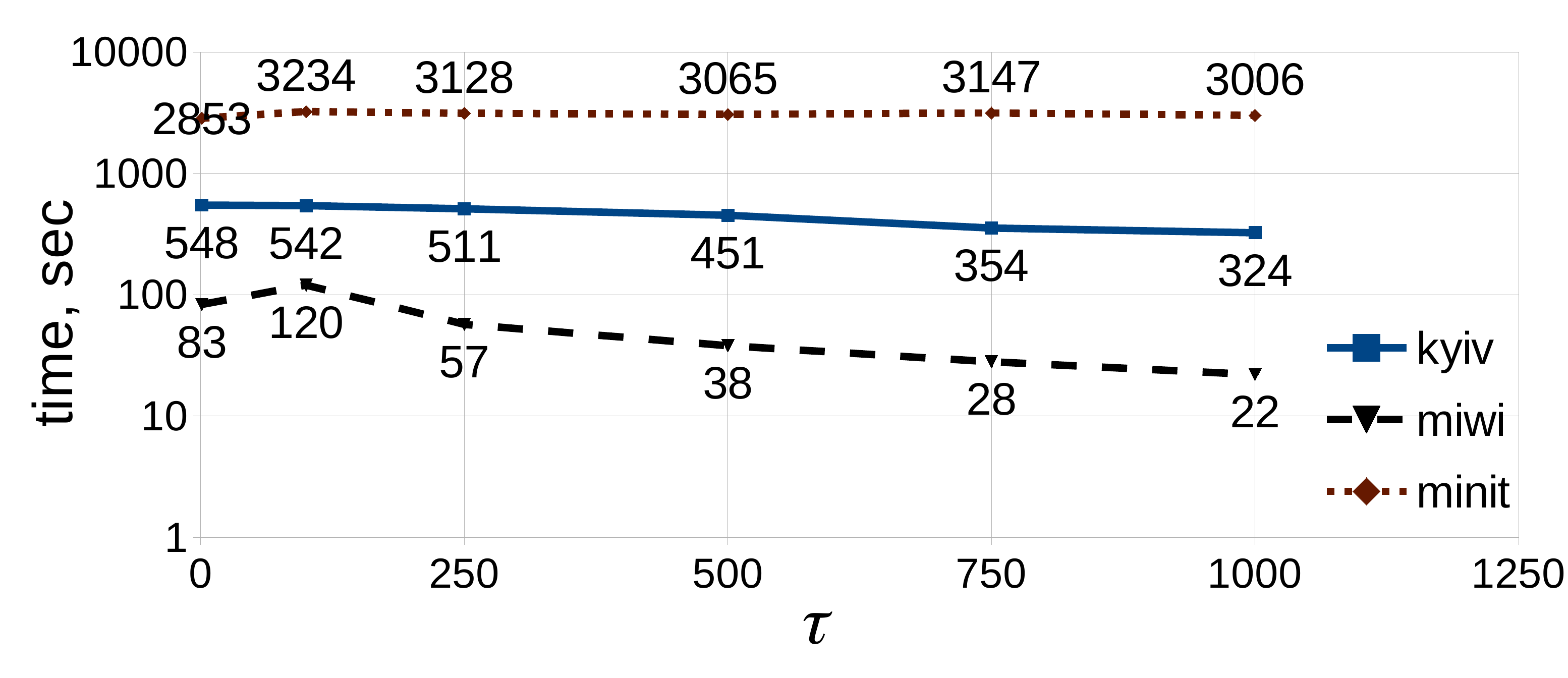}
      \label{fig24}
    }
    \subfloat[Pumsb, $k_{max} = 5$]{
      \includegraphics[width=0.485\textwidth]
      {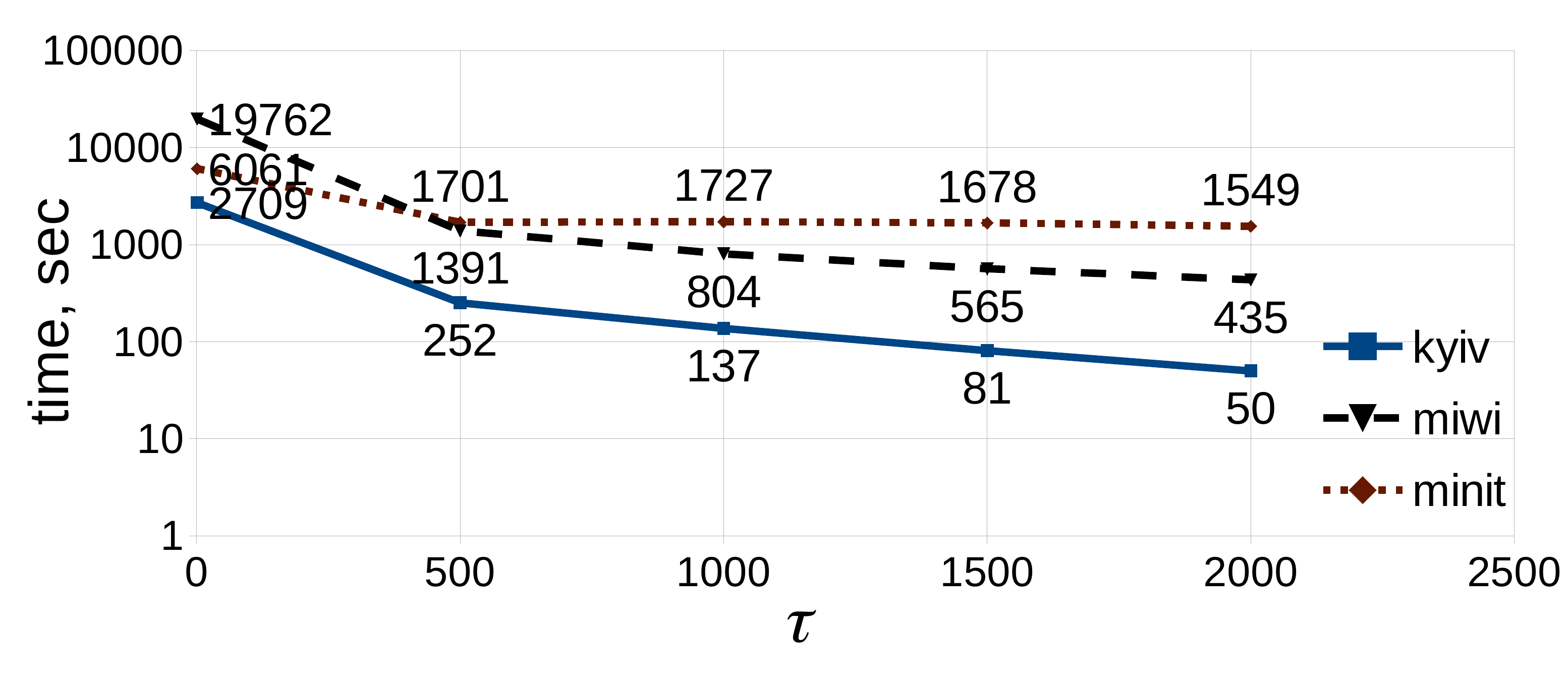}
      \label{fig25}
    }\\
    \subfloat[USCensus1990, $k_{max} = 3$, $\tau \in \{ 1, 250, 500, 750,
    1000, 2500,$ $5000, 10000 \}$]{
      \includegraphics[width=0.6\textwidth]
      {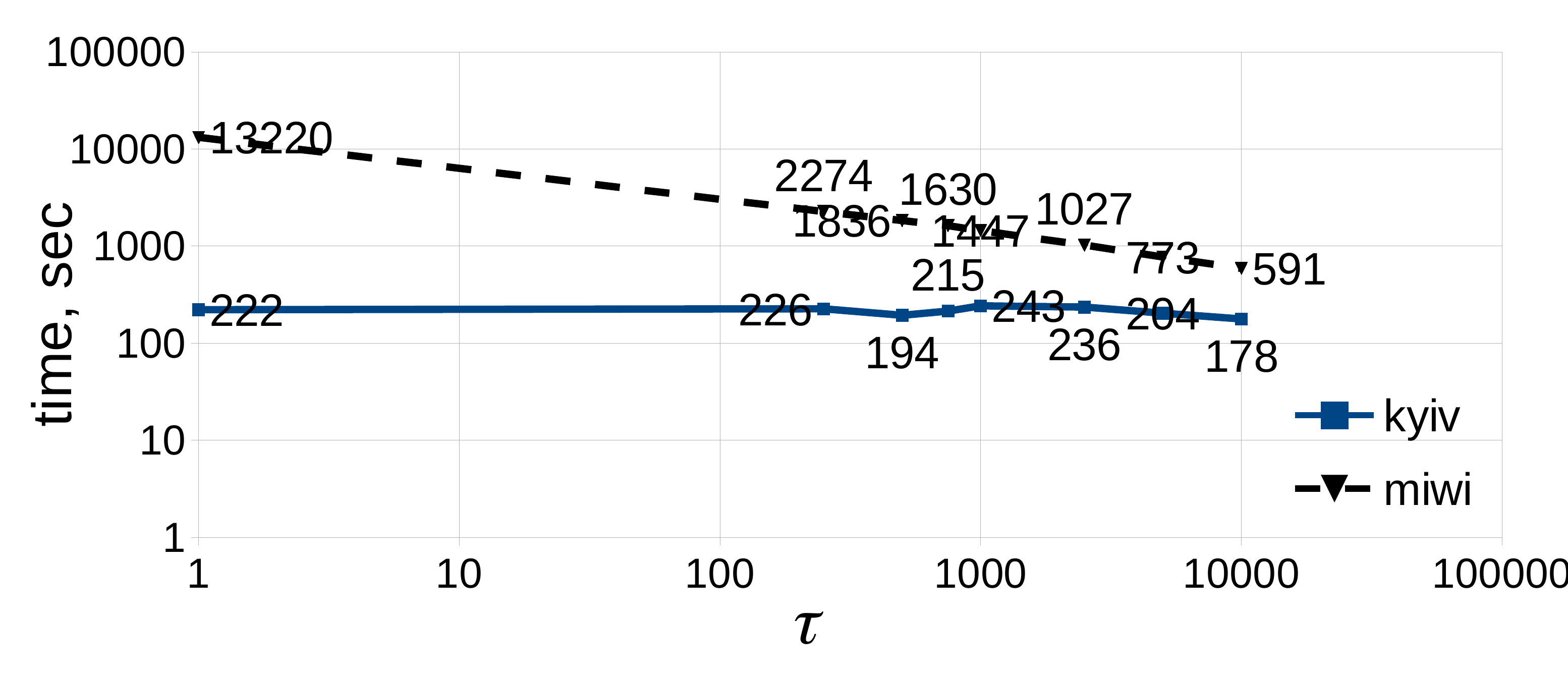}
      \label{fig27}
    }
    \captionof{figure}{Execution time vs $\tau$.}
    \label{fig24all}
  \end{figure}

  From Figures \ref{fig8all}, \ref{fig12all}, \ref{fig16all} and \ref{fig20all}
  it can be seen that the execution time of all algorithms tends to fall with
  increasing $\tau$. That is, finding minimal unique itemsets is more
  demanding that finding infrequent itemsets, as might be expected. This is
  studied in more detail in Figure \ref{fig24all} which plots the measured
  execution times vs $\tau$.

  It can be seen from Figure \ref{fig24} that MINIT's execution time initially
  increases with $\tau$ (see \cite{minit} where similar behaviour is reported),
  and then later falls as $\tau$ is increased further. Similarly, the execution
  time of MIWI also increases initially. We think that these initial increases
  are caused by the design of the algorithm and not by the dataset complexity
  since it is not present for Algorithm \ref{al:kyiv}.

  For this relatively simple dataset MIWI offers the shortest execution time.
  However, for the more complex Pumsb and USCensus1990 datasets it can be seen
  that Algorithm \ref{al:kyiv} offers the shortest execution time, although the
  performance gap between MIWI and Algorithm \ref{al:kyiv} narrows for large
  $\tau$ with the USCensus1990 dataset.

  To summarise, we conclude that Algorithm \ref{al:kyiv}'s execution time tends
  to decrease with $\tau$, its comparative performance with the MIWI and MINIT
  algorithms is approximately $\tau$-invariant and Algorithm \ref{al:kyiv}
  performs best when the input dataset is computationally expensive (such as
  the Pumsb or USCensus1990 datasets).

  \subsubsection{Memory Usage}

  \begin{figure}
    \centering
    \includegraphics[width=0.5\textwidth]
    {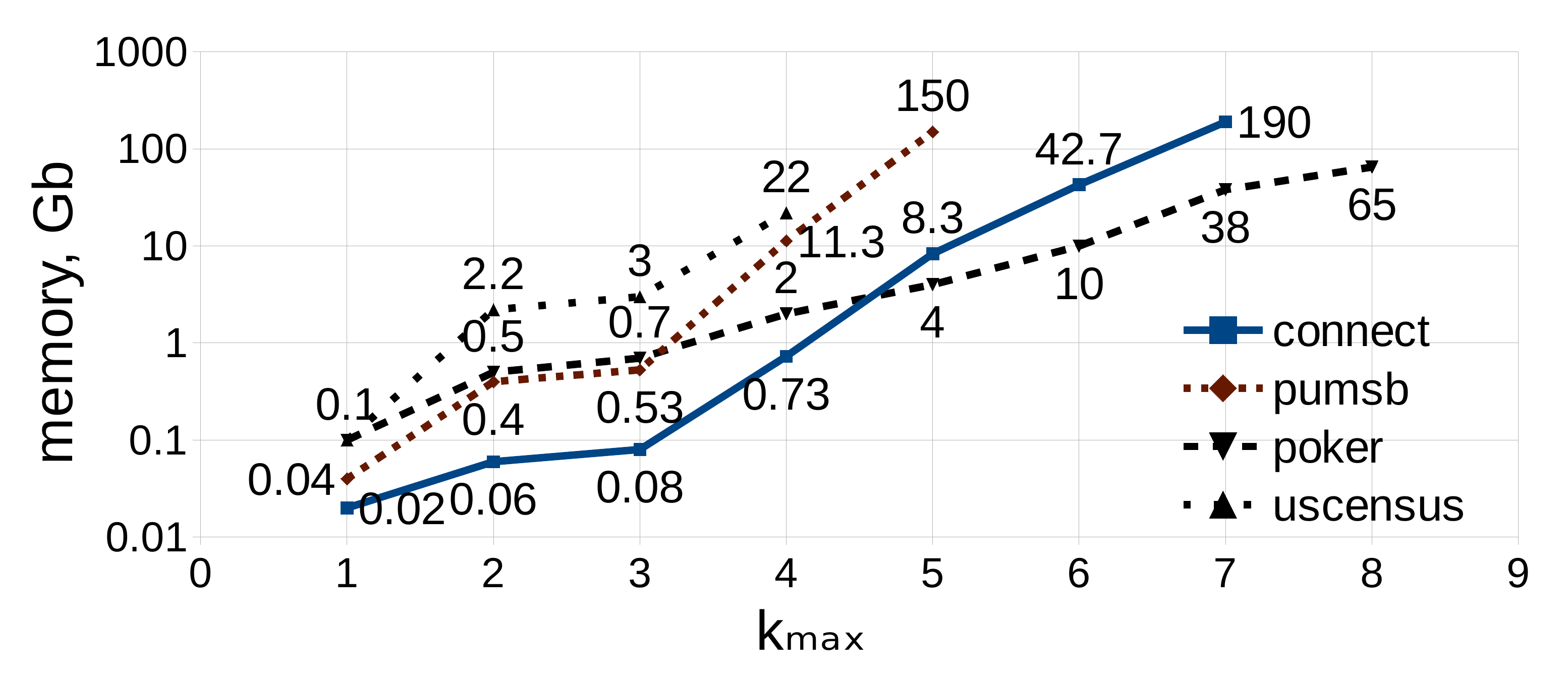}
    \captionof{figure}{Memory consumption of Algorithm \ref{al:kyiv} vs
    $k_{max}$, $\tau = 1$.}
    \label{fig28}
  \end{figure}

  Algorithm \ref{al:kyiv} intentionally trades increased memory for faster
  execution times via its use of a breadth-first approach. This is reasonable
  in view of the favourable scaling of memory size vs CPU speed on modern
  hardware. Figure \ref{fig28} shows the memory consumption of Algorithm
  \ref{al:kyiv} for the Connect, Pumsb, Poker and USCensus1990 datasets vs
  $k_{max}$. These plots indicate the maximum memory needed during algorithm
  execution and so this amount of memory ensures the fastest execution time
  since garbage collection is not required. For smaller amounts of memory the
  algorithm is observed to become somewhat slower as the Java Virtual Machine
  needs to start garbage collection.

  The memory requirement is dominated by storage of itemset rows to perform
  intersection. When $1 < k < k_{max}$, two levels of the prefix tree must be
  stored, but when $k = k_{max}$ (last level), then only one level needs to be
  stored (for example, the $190$Gb in Figure \ref{fig28} is mostly occupied by
  the $6$-itemset rows). Note that there is a level in the prefix tree that
  requires the largest amount of memory, a sort of equator. Above this value
  Algorithm \ref{al:kyiv} can compute all minimal unique itemsets without
  additional memory.

  \subsection{Parallel Algorithm Performance}

  \begin{figure}
    \centering
    \subfloat[Connect, $k_{max} = 6$, $\tau = 1$]{
      \includegraphics[width=0.485\textwidth]
      {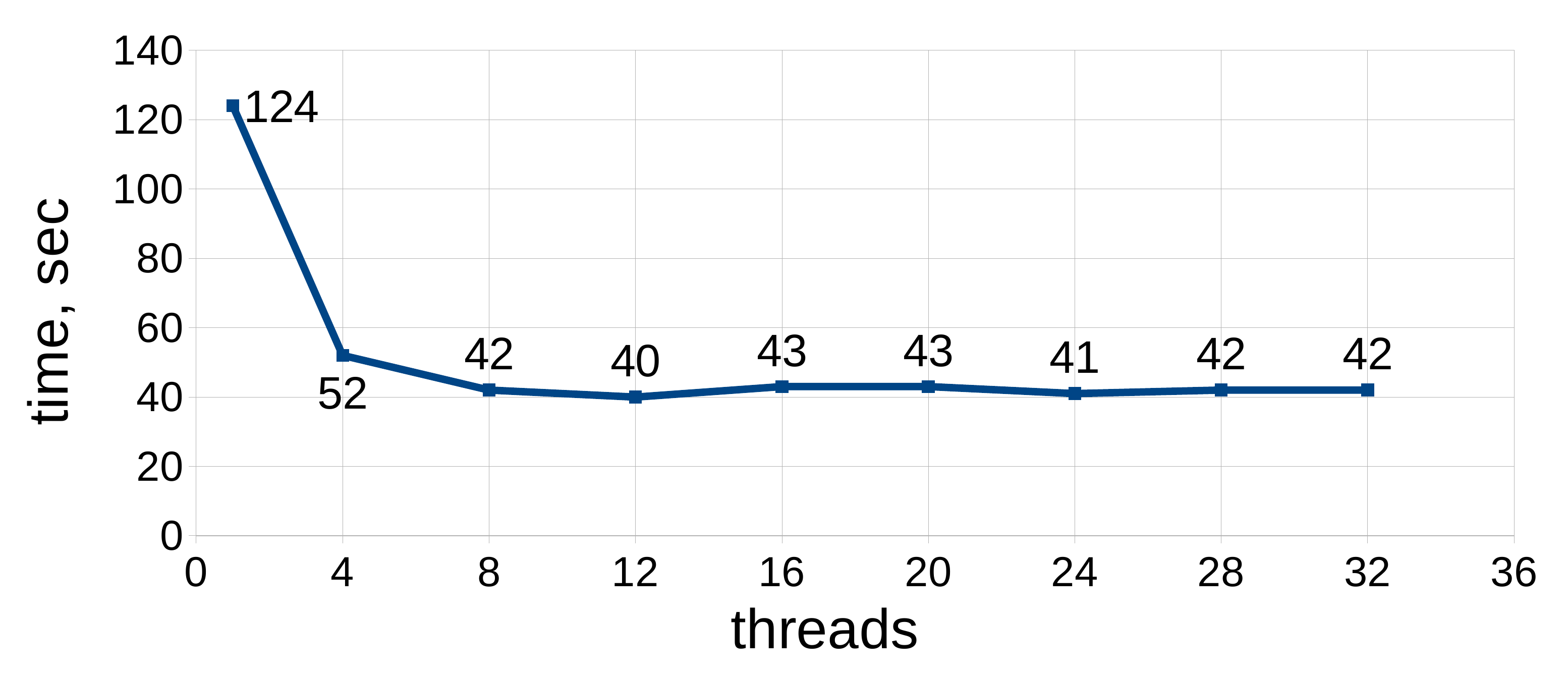}
      \label{fig29}
    }
    \subfloat[Pumsb, $k_{max} = 5$, $\tau = 1$]{
      \includegraphics[width=0.485\textwidth]
      {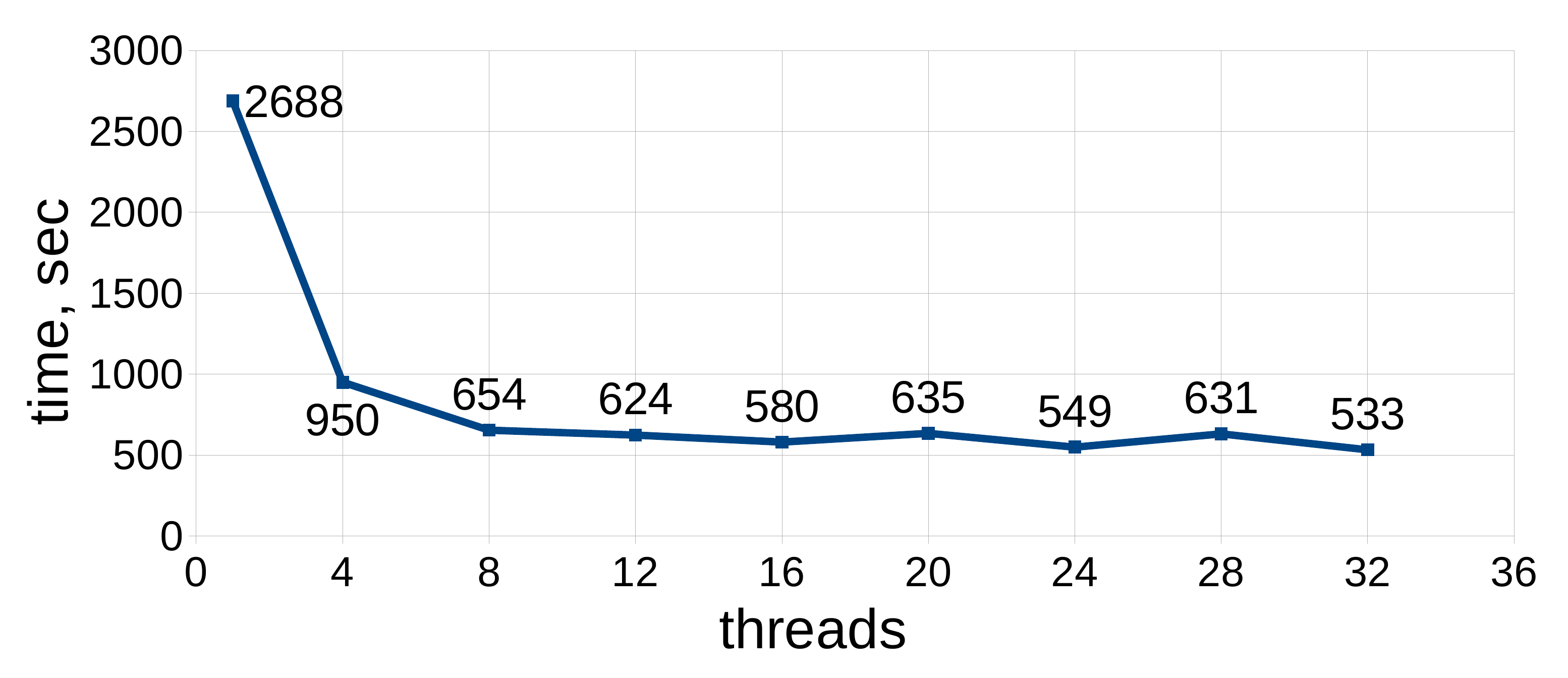}
      \label{fig30}
    }
    \captionof{figure}{Parallel algorithm execution time vs number of threads.}
  \end{figure}

  Figures \ref{fig29} and \ref{fig30} show execution time versus the number of
  threads used for the Connect and Pumsb datasets respectively. It can be seen
  that at around $8$ threads the performance saturates and additional threads
  yielding little further performance gain.

  \begin{table}[t]
    \centering
    \begin{tabular}{c c c c c}
    \hline \\ [-1ex]
    T & thread 1 & thread 2 & thread 3 & thread 4 \\ [1ex]
    \hline \\ [-1ex]
    & & $k = 3$ & & \\ [0.5ex]
    871 & 24 & 24 & 24 & 24 \\ [0.5ex]
    & & $k = 4$ & & \\ [0.5ex]
    871 & 340 & 344 & 343 & 342 \\ [0.5ex]
    & & $k = 5 = k_{max}$ & & \\ [0.5ex]
    871 & 468 & 501 & 470 & 482 \\ [0.5ex]
    \hline \\ [-1ex]
    \end{tabular}
    \caption{Granularity of $4$ threads for Pumsb, $k_{max} = 5$, $\tau = 1$.
    Time is given in seconds, levelwise. T column shows the whole execution
    time.}
    \label{tab1}
  \end{table}

  \begin{table}[t]
    \centering
    \begin{tabular}{c c c c c c c c c}
    \hline \\ [-1ex]
    T & t1 & t2 & t3 & t4 & t5 & t6 & t7 & t8 \\ [1ex]
    \hline \\ [-1ex]
    & & & & $k = 3$ & & & & \\ [0.5ex]
    674 & 21 & 17 & 19 & 21 & 21 & 21 & 19 & 21 \\ [0.5ex]
    & & & & $k = 4$ & & & & \\ [0.5ex]
    674 & 352 & 284 & 354 & 352 & 285 & 291 & 351 & 352 \\ [0.5ex]
    & & & & $k_{max}$ & & & & \\ [0.5ex]
    674 & 297 & 281 & 293 & 293 & 294 & 282 & 289 & 282 \\ [0.5ex]
    \hline \\ [-1ex]
    \end{tabular}
    \caption{Granularity of $8$ threads for Pumsb, $k_{max} = 5$, $\tau = 1$.
    Time is given in seconds, levelwise. T column shows the whole execution
    time.}
    \label{tab2}
  \end{table}

  \begin{table}[t]
    \centering
    \begin{tabular}{c c c c c c c c c}
    \hline \\ [-1ex]
    T & t1 & t2 & t3 & t4 & t5 & t6 & t7 & t8 \\ [1ex]
    \hline \\ [-1ex]
    & & & & $k = 3$ & & & & \\ [0.5ex]
    567 & 20 & 19 & 19 & 20 & 19 & 19 & 20 & 20 \\ [0.5ex]
    & & & & $k = 4$ & & & & \\ [0.5ex]
    567 & 342 & 345 & 258 & 345 & 342 & 333 & 260 & 346 \\ [0.5ex]
    & & & & $k_{max}$ & & & & \\ [0.5ex]
    567 & 178 & 171 & 177 & 171 & 170 & 170 & 179 & 179 \\ [0.5ex]
    \hline \\ [-1ex]
    T & t9 & t10 & t11 & t12 & t13 & t14 & t15 & t16 \\ [1ex]
    \hline \\ [-1ex]
    & & & & $k = 3$ & & & & \\ [0.5ex]
    567 & 20 & 19 & 19 & 19 & 20 & 19 & 19 & 19 \\ [0.5ex]
    & & & & $k = 4$ & & & & \\ [0.5ex]
    567 & 270 & 272 & 272 & 345 & 271 & 272 & 342 & 345 \\ [0.5ex]
    & & & & $k_{max}$ & & & & \\ [0.5ex]
    567 & 178 & 177 & 171 & 172 & 172 & 177 & 177 & 177 \\ [0.5ex]
    \hline \\ [-1ex]
    \end{tabular}
    \caption{Granularity of $16$ threads for Pumsb, $k_{max} = 5$, $\tau = 1$.
    Time is given in seconds, levelwise. T column shows the whole execution
    time.}
    \label{tab3}
  \end{table}

  In more detail, tables \ref{tab1}, \ref{tab2} and \ref{tab3} show the per
  thread execution times together with the overall execution time. Data is
  shown for $4$, $8$ and $16$ threads measured for the Pumsb dataset, $k_{max}
  = 5$, $\tau = 1$. It can be seen that the thread execution times consistently
  have a narrow spread, indicating that the workload is divided evenly amongst
  the threads. That is, there is not one slow thread which dominates parallel
  execution time. Observe also that the execution times in the last row of each
  table (when $k = 5 = k_{max}$) decrease as the number of threads is increased
  but that the maximum thread execution times when $k = 3$ and $k = 4$ do not
  show a similar decrease. This may be due to the communication overhead when
  transitioning between layers in the search tree, although we leave detailed
  analysis of this to future work.

\section{Summary and Conclusions} \label{sec:concl}

  A new algorithm for finding quasi-identifiers within a data set is
  introduced, where a quasi-identifier is a subset of attributes that can
  uniquely identify data set records (or identify that a record lied within a
  small group of $\tau$ records). This algorithm is demonstrated to be
  substantially faster than the state of the art, to scale well to large data
  sets and to be amenable to parallelisation with well-balanced thread
  execution times.

  \subsection{Further Improvements and Optimisation}

  We briefly highlight areas where further efficiency gains may be possible,
  although we leave these as future work.

  Regarding memory usage, suppose Kyiv that is able to compute the
  $k^*$-itemsets by intersecting the $(k^* - 1)$-itemsets but that the
  algorithm goes out of memory at the $k^* + 1$ level. We might keep
  intersecting the $(k^* - 1)$-itemsets in order to find not only the
  $k^*$-itemsets, but also the $(k^* + \delta)$-itemsets, where $\delta \in
  \mathbb{N}$ at each consecutive level of the prefix tree. This would allow us
  to halt growth in memory usage as this is mainly used for itemset storage.
  Related technical refinements could be to implement the corresponding itemset
  test using the $(k^* - 1)$-itemsets and to use data compression for the array
  storage to decrease the memory consumption, albeit at the cost of increased
  execution time.

  Regarding data structures, it would be useful to get a better understanding
  of the most efficient structures for storing the prefix tree and handling the
  search space operations. The insights gained might improve the parallel form
  of the algorithm. One possible direction would be to look at an array
  implementation of a tree structure representation, e.g. similar to the work
  in \cite{fpgrowth2}.

  The main computational bottleneck, the intersection operation, could
  potentially be improved by making use of the specialised SSE (Streaming SIMD
  Extensions) instructions available on Intel processors. There exists
  performance analysis \cite{intelsse} indicating that use of these
  instructions might produce a $4\times$ speed up.

  \bibliographystyle{ACM-Reference-Format-Journals}
  \bibliography{kyiv}

\end{document}